\documentclass[12pt,draftcls,onecolumn]{IEEEtran}

\usepackage{tabularx}

\usepackage{algorithmic,algorithm}
\usepackage{ae}
\usepackage[T1]{fontenc}
\usepackage{epsf}
\usepackage{epsfig}
\usepackage{amsmath}
\usepackage{amsfonts}
\usepackage{amssymb}
\usepackage{amsxtra}
\usepackage{latexsym}
\usepackage{color}
\usepackage{subfigure}
\usepackage{graphics}
\usepackage{algorithm}
\usepackage{algorithmic}
\usepackage{mathrsfs}
\usepackage{cite}
\usepackage{url}
\usepackage{ulem}

\usepackage[english]{babel}
\usepackage{caption} 

\usepackage[section]{placeins}

\begin{document} 

\title{Low-Coherence Frames from Group Fourier Matrices }
\author{Matthew Thill \hspace{12pt} Babak Hassibi \\ 
Department of Electrical Engineering, Caltech, Pasadena, CA} 
\maketitle 

\newtheorem{definition}{Definition}
\newtheorem{theorem}{Theorem}
\newtheorem{lemma}{Lemma}
\newtheorem{corollary}{Corollary}

\newcommand{\stexp}{\mbox{$\mathbb{E}$}}    
\newcommand{\beq}{\begin{equation}}
\newcommand{\eeq}{\end{equation}}
\newcommand{\bea}{\begin{eqnarray}}
\newcommand{\eea}{\end{eqnarray}}
\def\G{{\mathcal{G}}}
\newcommand{\Prob}{\ensuremath{\mathbb{P}}}

\newcommand{\matr}[1]{\mathbf{#1}}
\newcommand{\vect}[1]{\mathbf{#1}}

\newcommand{\ve}{\vect{e}}
\newcommand{\hve}{\vect{\hat e}}
\newcommand{\vnu}{\vect{\nu}}

\newcommand{\norm}[2]{\lVert #1 \rVert_{#2}}
\newcommand{\setweightR}[2]{\Sigma_{\R^{#1}}^{(#2)}} 

\newcommand{\eg}{\emph{e.g.}}
\newcommand{\ie}{\emph{i.e.}}
\newcommand{\etal}{\emph{et al.}}
\newcommand{\confer}{\emph{cf.}}

\newcommand{\Q}{\mathbb{Q}} 
\newcommand{\R}{\mathbb{R}} 
\newcommand{\Z}{\mathbb{Z}} 
\newcommand{\C}{\mathbb{C}} 
\newcommand{\F}{\mathbb{F}}
\newcommand{\vv}{\vect{v}} 
\newcommand{\U}{\matr{U}} 
\newcommand{\M}{\matr{M}} 
\newcommand{\x}{\vect{x}}
\newcommand{\y}{\vect{y}} 
\newcommand{\cc}{\vect{c}} 
\newcommand{\va}{\vect{a}} 
\newcommand{\w}{\vect{w}} 
\newcommand{\Tr}[1]{\text{Tr}(#1)}

\begin{abstract} 
Many problems in areas such as compressive sensing and coding theory seek to design a set of equal-norm vectors with large angular separation. This idea is essentially equivalent to constructing a frame with low coherence. The elements of such frames can in turn be used to build high-performance spherical codes, quantum measurement operators, and compressive sensing measurement matrices, to name a few applications. 

In this work, we allude to the group-frame construction first described by Slepian and further explored in the works of Vale and Waldron. We present a method for selecting representations of a finite group to construct a group frame that achieves low coherence. Our technique produces a tight frame with a small number of distinct inner product values between the frame elements, in a sense approximating a Grassmanian frame. We identify special cases in which our construction yields some previously-known frames with optimal coherence meeting the Welch lower bound, and other cases in which the entries of our frame vectors come from small alphabets. In particular, we apply our technique to the problem choosing a subset of rows of a Hadamard matrix so that the resulting columns form a low-coherence frame. Finally, we give an explicit calculation of the average coherence of our frames, and find regimes in which they satisfy the Strong Coherence Property described by Mixon, Bajwa, and Calderbank. 
\end{abstract} 


\section{Introduction:\\ Matrix Coherence and Frames} 
\label{sec:matcohframe} 

Recall that a set of vectors $\{f_i\}_{i = 1}^n$ in $\C^m$ (respectively $\R^m$) forms a \textit{frame} if there are two positive constants $A$ and $B$ such that for any $f \in \C^m$ (respectively $\R^m$), we have 
\begin{align} 
A || f ||^2 \le \sum_{i = 1}^n | \langle f, f_i \rangle |^2 \le B || f ||^2. \label{eqn:framedefeqn} 
\end{align} 

If we arrange our frame elements to be the columns of a matrix $\M = [f_1, ..., f_n]$, the frame is said to be \textit{tight} if $\M \M^* = \lambda \textbf{I}$, where $\lambda$ is a real scalar and $\textbf{I}$ is the $m \times m$ identity matrix. This is equivalent to having $A = B$ in (\ref{eqn:framedefeqn}), and we can see that in this case we necessarily have $A = B = \lambda$. The frame is called \textit{unit norm} if $||f_i|| = 1$, $\forall i$, and a tight unit-norm frame satisfies $\M \M^* = \frac{n}{m} \textbf{I} \in \C^{m \times m}$. We will typically restrict our attention to unit-norm frames. 

Of particular interest in frame design are the magnitudes of the inner products between distinct frame elements, $| \langle f_i, f_j \rangle |$, $i \ne j$. A unit-norm frame is called \textit{equiangular} if all of these magnitudes are equal: $| \langle f_i, f_j \rangle | = \alpha$, $\forall i \ne j$, for some constant $\alpha$. In general, we would like all of the inner product magnitudes to be as small as possible so that the frame vectors are well-spaced about the $m$-dimensional unit sphere. Such frames popularly have applications to coding theory (for instance, spherical codes \cite{Delsarte, Heath} and LDPC codes \cite{Fan}) and compressive sensing \cite{Donoho, Huo, Candes, Romberg, Tao, Tropp07}. They also arise in areas such as quantum measurements \cite{Eldar, Scott, RenesBlumeKohoutScottCaves} and MIMO communications \cite{Bolcskei, Pailraj}. Recently, frame theory itself has proven to be an exciting field, and has been notably studied by Casazza, Kutyniok, Fickus, Dixon, and others \cite{CasazzaKutyniokFickusMixon, FickusMixonTremain, CasazzaArt, CasazzaKutyniok2, CasazzaKovacevic, CasazzaKutyniokLi}. 

On this note, let us formalize our particular problem: We would like to minimize the \textit{coherence} $\mu$ of a unit-norm frame defined as the largest inner product magnitude between two distinct frame elements: 
\begin{align} 
\mu := \max_{i \ne j} |\langle f_i, f_j \rangle |. 
\end{align} 
The following is a classical lower bound on the coherence due to Welch \cite{Welch}: 
\begin{theorem} 
\label{thm:framebound} 
Let $\{f_i\}_{i = 1}^n$ be a unit-norm frame in $\C^m$ or $\R^m$. The coherence $\mu := \max_{i \ne j} |\langle f_i, f_j \rangle |$ satisfies 
\begin{align} 
\mu \ge \sqrt{\frac{n-m}{m(n-1)}}, \label{eqn:welchbound} 
\end{align} 
with equality if and only $\{f_i\}$ is both tight and equiangular. In this case, the frame is called \textit{Grassmanian}. 
\end{theorem} 

\begin{proof} 
The bound in (\ref{eqn:welchbound}) is one of a more general set of bounds originally derived by Welch in \cite{Welch}. This version of the theorem is typically proven (e.g. in \cite{Heath}) by considering the eigenvalues of the Gram matrix $\matr{G} := [ \langle f_i, f_j \rangle]$. 
\end{proof} 

It should be emphasized that the Welch bound is often quite difficult to achieve with equality, and in fact for certain values of $n$ and $m$ it is impossible to achieve it. So in general, we would like to construct frames which have coherence close to the Welch bound. 

If we omit the requirement that our frame be equiangular, we obtain the following relaxation of the Welch bound: 
\begin{lemma} 
Let $\{f_i\}_{i = 1}^n$ be a unit-norm frame in $\C^m$ or $\R^m$. Then the mean value of the $n(n-1)$ squared inner product norms $\{|\langle f_i, f_j \rangle |^2 \}_{i \ne j}$ satisfies 
\begin{align} 
\frac{1}{n(n-1)} \sum_{i \ne j} |\langle f_i, f_j \rangle |^2 \ge \frac{n-m}{m(n-1)}, \label{eqn:tightframevareqn} 
\end{align} 
with equality if and only if $\{f_i\}$ is a tight frame. 
\label{lem:tightframevar} 
\end{lemma} 

\begin{proof} 
The quantity $\frac{1}{n(n-1)} \sum_{i \ne j} |\langle f_i, f_j \rangle |^2$ is very closely related to the \textit{frame potential} defined in \cite{BenedettoFickus}, and (\ref{eqn:tightframevareqn}) follows from Theorem 6.2 in that work. 
\end{proof} 

Lemma \ref{lem:tightframevar} allows us to obtain upper bounds on the coherence of tight frames, which become particularly effective when there are few distinct inner product values $\{|\langle f_i, f_j \rangle |\}_{i \ne j}$, with each value arising the same number of times in this set. 

\begin{lemma} 

Let $\{f_i\}_{i = 1}^n$ be a unit-norm tight frame in $\C^m$ or $\R^m$, such that the inner product norms $\{|\langle f_i, f_j \rangle |\}_{i \ne j}$ take on $\kappa$ distinct values, with each value arising the same number of times as such an inner product norm. Then the coherence $\mu$ of the frame is at most a factor of $\sqrt{\kappa}$ greater than the Welch bound: 
\begin{align} 
\mu & \le \sqrt{\kappa} \sqrt{\frac{n-m}{m(n-1)}}. \label{eqn:tightframecohbound} 
\end{align} 

\label{lem:tightframecohbound} 
\end{lemma} 

\begin{proof} 
Consider the squares of the $\kappa$ distinct inner product values. Since each of these squares arises the same number of times as a squared norm $|\langle f_i, f_j \rangle|^2$, we have that the average $\frac{1}{n(n-1)}\sum_{i \ne j} |\langle f_i, f_j \rangle |^2$ is equal to the mean of the $\kappa$ squares. By Lemma \ref{lem:tightframevar}, this mean is equal to $\frac{n-m}{m(n-1)}$. This means that none of the squares can be larger than $\kappa \frac{n-m}{m(n-1)}$, so the coherence is bounded by $\sqrt{\kappa} \sqrt{\frac{n-m}{m(n-1)}}$. 
\end{proof} 

When $\kappa = 1$ in Lemma \ref{lem:tightframecohbound}, the frame becomes both tight and equiangular, and fittingly the coherence in (\ref{eqn:tightframecohbound}) achieves the Welch bound. In some of our previous work \cite{ThillAllerton}, we provided constructions of unit-norm tight frames in which we could control the number of distinct inner product values and ensure that each would arise with the same multiplicity. 

In what follows, we will generalize our original construction using a framework which also encompasses some of our subsequent results on frames arising from representations of various groups \cite{Thill_ICASSP2014, Thill_ICASSP2013, Thill_ISIT2013}. Frames constructed in this fashion are called ``group frames,'' and will be reviewed in Section \ref{sec:groupframes}. In Section \ref{sec:groupfouriermat}, we will discuss the connection between tight group frames and generalized group Fourier matrices. The problem of designing tight group frames with low coherence is essentially that of choosing a subset of rows of one of these matrices to form a matrix with low coherence.  Viewing the problem in this manner allows us to use character theory to design our frames, and leads us to a frame construction that we present in Section \ref{sec:choosingautomorphisms} along with bounds on the resulting frame coherence. We explore how this construction performs when using representations of several different types of groups, including general linear groups, vector spaces over finite fields, and special linear groups. We will present methods to maintain low alphabet sizes in our frame entries, sometimes producing frames which can be realized as a subset of rows of a Hadamard matrix (see Section \ref{sec:HadamardFrames}). In certain cases, we end up with frames arising from difference sets over finite fields, so in some sense we give insight into the frames described in \cite{Giannakis} and \cite{Ding}. Finally, in Section \ref{sec:strongcoherence} we examine when these frames satisfy the Coherence Property and the Strong Coherence Property described by \cite{BajwaCalderbankJafarpour} and \cite{MixonBajwaCalderbank}. In particular, we utilize the group structure of our frames to calculate their \textit{average coherence} explicitly.

\section{Group Frames} 
\label{sec:groupframes} 

One technique first used by Slepian \cite{Slepian} in constructing his ``group codes'' to reduce the number of inner product values between vectors involves taking the images of a unit-norm vector $\vv \in \C^{m \times 1}$ under a finite multiplicative group of unitary matrices $\mathcal{U} = \{\U_1, ..., \U_n\} \subset \C^{m \times m}$. The resulting set of vectors $\{\U_1 \vv, ..., \U_n \vv\}$ has mutual inner products in the form 
\begin{align} 
\vv^* \U_i^* \U_j \vv &= \vv^* \U_i^{-1} \U_j \vv = \vv \U_k \vv, \label{eqn:mutualinnerproducts} 
\end{align} 
where we have exploited the fact that $\U_i^{-1} \U_j$ is some other element $\U_k$ of $\mathcal{U}$. This immediately reduces the number of inner product values that can arise from a possible $n \choose 2$ to only at most $n-1$, corresponding to each of the nonidentity group elements of $\mathcal{U}$. (The identity element $\U_1$ is simply the identity matrix $\matr{I} \in \C^{m \times m}$. The corresponding inner product $\vv^* \vv = 1$ arises only when taking the inner product of a vector $\U_i \vv$ with itself, and is thus not considered in calculating the coherence of the frame.) 

One can quickly verify that each of the values $\vv^* \U_k \vv$ arises the same number of times as the inner product between two vectors $\U_i \vv$ and $\U_j \vv$ (since for each $\U_i$ there is exactly one element $\U_j$ that satisfies $\U_i^* \U_j = \U_k$, and it is given by $\U_j := \U_i \U_k$). As a result, the group codes naturally lend themselves to analysis using Lemmas \ref{lem:tightframevar} and \ref{lem:tightframecohbound}. 

Frames arising in this form are called ``group frames,'' and are well-studied in the works of Vale, Waldron, and others \cite{Vale, ValeWaldron2, ValeWaldron3, Waldron1, ChienWaldron, HayWaldron}. A great review of group frames is provided in \cite{CasazzaKutyniok}. 

When $\mathcal{U}$ is chosen to be abelian, so that all the $\U_i$ commute with each other, then the matrices can be simultaneously diagonalized by a unitary change of basis matrix $\matr{B}$ so that we may write $\matr{B}^* \U_i \matr{B} = \matr{D}_i$, where $\matr{D}_i$ is diagonal. In this case the inner product corresponding to $\U_i$ will take the form $$\vv^*\U_i \vv = \vv^* \matr{B}^* \matr{D}_i \matr{B} \vv, $$ so by replacing $\vv$ with $\matr{B}^* \vv$ without loss of generality, we may assume that the $\U_i$ are already diagonal. Furthermore, since each $\U_i$ must have a multiplicative order dividing the size of $\mathcal{U}$, we may take the diagonal entries of $\U_i$ to be powers of the $n^{th}$-root of unity $\omega := e^\frac{2 \pi i}{n}$. The matrices will then take the form $$\U_j = \text{diag}(\omega^{a_{1, j}}, ..., \omega^{a_{m, j}}) \in \C^{m \times m}, $$ where the $a_{i, j}$ are integers between $0$ and $n-1$. In the language of representation theory, we have decomposed $\mathcal{U}$ into its degree-1 irreducible representations. 

If we write the coordinates of our rotated vector as $\vv = (v_1, ..., v_m)^T \in \C^{m \times 1}$, then our inner products will now take the form 
\begin{align} 
\vv^* \U_j \vv & = \sum_{i = 1}^m \omega^{a_{i, j}} |v_i|^2, \label{eqn:abelianinnerproducts}  
\end{align} 
so we see that the inner products depend only on the magnitudes of the $v_i$, which weight the diagonal entries of the $\U_j$. A natural choice is therefore to choose all the $v_i$ to be real and of equal magnitude (i.e., take $\vv := \frac{1}{\sqrt{m}} [1, ..., 1]^T \in \C^{m \times 1}$). When we do this, and in addition require the sets of diagonal components $\{\omega^{a_{i, j}}\}_{i = 1}^m$ to form distinct representations of $\mathcal{U}$, we obtain what is called a \textit{harmonic frame}, which we will define concretely as follows: 

\begin{definition} 
Let $m$ and $n$ be integers, $\omega = e^\frac{2 \pi i}{n}$, and $\U_j = \text{diag}(\omega^{a_{1, j}}, ..., \omega^{a_{m, j}}) \in \C^{m \times m}$ for $j = 1, ..., n$, where the $a_{i, j}$ are integers between $0$ and $n-1$. If we set $\vv = \frac{1}{\sqrt{m}} [1, ..., 1]^T \in \C^{m \times 1}$, and $\M = [\U_1 \vv, ..., \U_n \vv]$, then if the rows of $\M$ are distinct, we call the set of columns $\{\U_j \vv\}_{j = 1}^n$ a \textit{harmonic frame}. 
\end{definition} 

Harmonic frames are one of the most thoroughly-studied types of structured frames \cite{ChienWaldron, HayWaldron}. As a preliminary result, we have the following: 

\begin{lemma} 
A harmonic frame is a tight, unit-norm frame. 
\end{lemma} 

\begin{proof} 
The fact that harmonic frames are unit-norm follows straight from the definition. The rest of this lemma is proven in \cite{CasazzaKutyniok}, and we will also explain the tightness of harmonic frames in Section \ref{sec:groupfouriermat}. 
\end{proof}

An important example of a harmonic frame arises when we choose the group $\mathcal{U}$ to be cyclic, meaning that each $\U_j$ is a power of a single matrix $\U = \text{diag}(\omega^{a_1}, ..., \omega^{a_m})$, so we may write $\U_j := \U^j$. In this case, if we again take $\vv$ to be the normalized vector of all 1s, our frame matrix takes the form 
\begin{align} 
\M & = \begin{bmatrix}  \vv & \U \vv & \hdots & \U^{n-1} \vv \end{bmatrix} \\  
 & =  \frac{1}{\sqrt{m}} \begin{bmatrix} 1 &  \omega^{a_{1}} & \omega^{a_{1}\cdot 2} & \hdots & \omega^{a_{1}\cdot (n-1)}\\ 
						1 & \omega^{a_{2}} & \omega^{a_{2} \cdot 2} & \hdots & \omega^{a_{2}\cdot (n-1)}\\ 
						\vdots & \vdots & \vdots & \ddots & \vdots \\ 
						1& \omega^{a_{m}} & \omega^{a_{m} \cdot 2} & \hdots & \omega^{a_{m} \cdot (n-1)} \end{bmatrix},  
\label{eqn:Meqn} 
\end{align} 
where the columns form a harmonic frame precisely when the $a_i$ are distinct. In this form, we see that $\M$ is a subset of rows of the $n \times n$ discrete Fourier matrix, so it becomes clear that the frame is tight since $\M \M^* = \frac{n}{m} \matr{I} \in \C^{m \times m}$. The question now becomes how to choose the frequencies $a_i$ to produce frames with low coherence?

In our previous work \cite{ThillAllerton, Thill_ICASSP2013}, we developed a method to obtain a range of frames with few distinct inner product values when $\mathcal{U}$ is a cyclic group and $\vv$ is the normalized vector of all 1s. We presented resulting upper bounds on the coherence of our frames which came reasonably close to the Welch bound. Our method was simple: choose $n$ to be a prime so that the set of integers $\{1, 2, ..., n-1\}$ forms a cyclic group under multiplication modulo $n$, denoted $(\Z/n\Z)^\times$. Then choose $m$ to be a divisor of $n-1$. Since $(\Z/n\Z)^\times$ is cyclic, it has a unique cyclic subgroup $a$ of size $m$. We choose the $a_i$ to be the integer elements of this subgroup. 

\begin{theorem} 
Let $n$ be a prime, $m$ a divisor of $n-1$, and $A = \{a_i\}_{i = 1}^m$ the unique subgroup of $(\Z/n\Z)^\times$ of size $m$. Set $\omega = e^{\frac{2 \pi i}{n}}$, $\vect{v} = \frac{1}{\sqrt{m}}[1, ..., 1]^T \in \R^{m}$, and $\matr{U} = \text{diag}(\omega^{a_1}, ..., \omega^{a_m})$. Then the inner products between the vectors $\{\U^\ell \vv\}_{\ell = 0}^{n-1}$ take on only $\frac{n-1}{m}$ values (possibly non-unique), each value occurring with the same number of times as one  of these mutual inner products. 
\label{thm:cosetinnerproducts} 
\end{theorem} 

\begin{proof} 
This proof appears in \cite{Thill_GroupMatrixCoherence}. The idea is that the inner product corresponding to the element $\U^\ell$ in (\ref{eqn:abelianinnerproducts}) is 
\begin{align} 
\vv^* \U^\ell \vv &= \sum_{i = 1}^n \omega^{\ell a_i}, \label{eqn:ellAsum} 
\end{align} 
and we can see that this sum depends only on the \textit{coset} of $A$ in which $\ell$ lies in the group $(\Z/n\Z)^\times$. There is only one inner product value for each of the $\frac{n-1}{m}$ cosets, and it is not too difficult to see that each of these values arises the same number of times as an inner product $\vv^* \U^{-\ell_1} \U^{\ell_2} \vv$ between two frame vectors $\U^{\ell_1} \vv$ and $\U^{\ell_2} \vv$. 
\end{proof}

\begin{figure}
\begin{center}
 \includegraphics[height = 230pt, width=230pt]{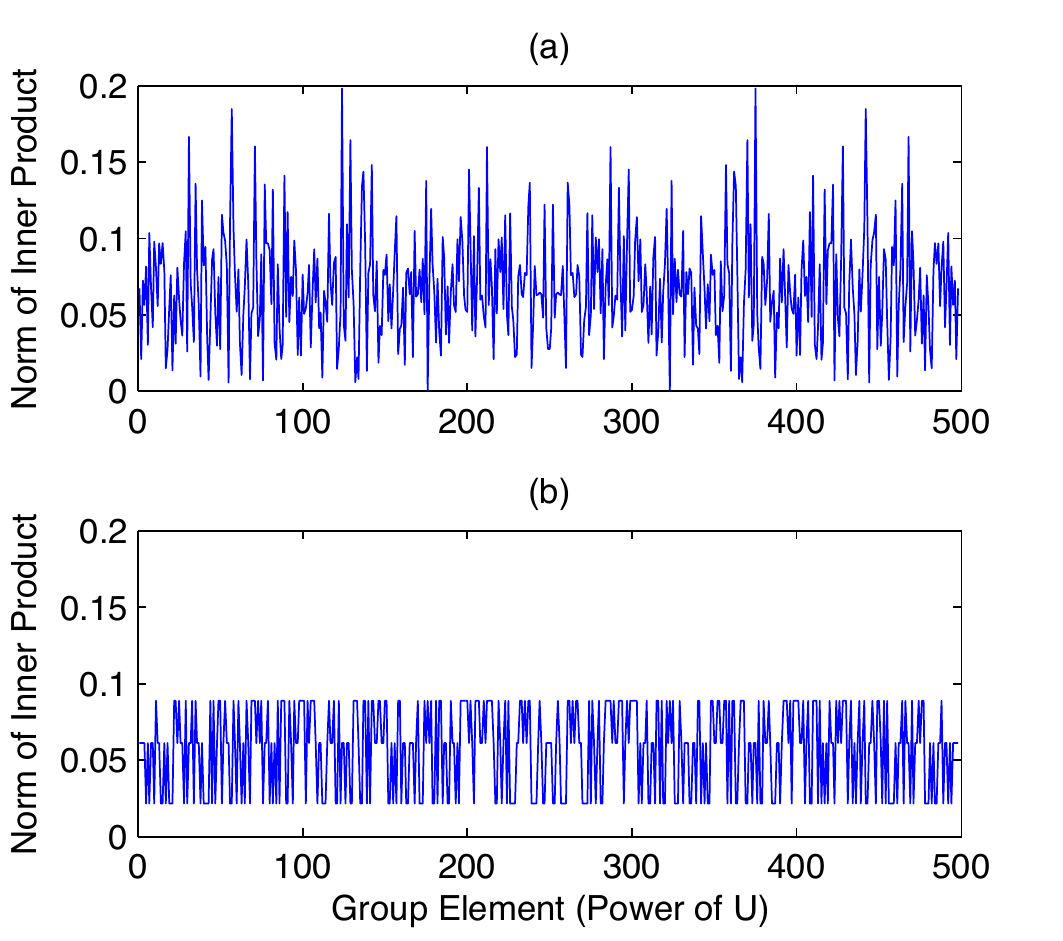}
 \caption{The norms of the inner products between the frame elements of $\M$ in (\ref{eqn:Meqn}). Here, $n = 499$ and $m = 166$. There is one inner product for each element $U^\ell$ of the cyclic group $\langle U \rangle \cong (\Z/n\Z)^\times$, where $n$ is prime. (a) The exponents $A = \{a_i\}_{i = 1}^m$ in (\ref{eqn:Meqn}) are chosen randomly. (b) $A$ is chosen to be the unique subgroup of $(\Z/n\Z)^\times$ of index 3. In this case, there are only 3 distinct inner product norms---one for each coset of $A$---each arising the same number of times. This leads to lower coherence. }
 \label{fig:innerprod_499_166}
\end{center}
\end{figure}

Figure \ref{fig:innerprod_499_166} depicts how choosing the exponents in (\ref{eqn:Meqn}) according to Theorem \ref{thm:cosetinnerproducts} reduces the number of distinct inner product values and yields low coherence frames. This theorem allows us to bound the coherence of our frames using Lemma \ref{lem:tightframecohbound}, but it turns out that we can achieve even tighter bounds by exploiting the algebraic structure employed in constructing our frames. In \cite{Thill_ICASSP2013}, we presented the following bounds: 

\begin{theorem}[General $\kappa$] 
Let $n$ be a prime, $m$ a divisor of $n-1$, and $\omega = e^{\frac{2 \pi i}{n}}$. Let $A = \{a_1, ..., a_m\}$ be the unique subgroup of $(\Z/n\Z)^\times$ of size $m$, and set $\U = diag(\omega^{a_1}, ..., \omega^{a_m}) \in \C^{m \times m}$, $\vv = \frac{1}{\sqrt{m}} [1, ..., 1]^T \in \C^{m \times 1}$, and $\matr{M} = [\vv, \U \vv, ..., \U^{n-1} \vv]$. 

If $\kappa := \frac{n-1}{m}$, then the coherence $\mu$ of $\matr{M}$ satisfies the following upper bound: 
\begin{align} 
\mu & \le \frac{1}{\kappa} \left( (\kappa-1) \sqrt{\frac{1}{m} \left(\kappa + \frac{1}{m}\right)} + \frac{1}{m} \right). 
\end{align} 
\label{thm:coherenceupperbound} 
\end{theorem} 

\begin{proof} 
We will actually realize this result in a slightly more general context, and will prove it in Appendix \ref{sec:generalkappa}. 
\end{proof} 

Interestingly, we can find an even tighter bound for the coherence of our frames in the case when $m$ is odd: 

\begin{theorem}[$m$ odd] 
Let $n$ be an odd prime, $m$ a divisor of $n-1$, and $\omega = e^{\frac{2 \pi i}{n}}$. Let $A = \{a_1, ..., a_m\}$ be the unique subgroup of $(\Z/n\Z)^\times$ of size $m$, and set $\U = {\rm diag}(\omega^{a_1}, ..., \omega^{a_m}) \in \C^{m \times m}$, $\vv = \frac{1}{\sqrt{m}} [1, ..., 1]^T \in \C^{m \times 1}$, and $\matr{M} = [\vv, \U \vv, ..., \U^{n-1} \vv]$. Set $\kappa := \frac{n-1}{m}$. 

If $m$ is odd, then the coherence of $\matr{M}$ is upper-bounded by 
\begin{align} 
\mu & \le \frac{1}{\kappa} \sqrt{\left(\frac{1}{m} + \left(\frac{\kappa}{2} - 1 \right) \beta \right)^2 + \left(\frac{\kappa}{2}\right)^2 \beta^2},  
\end{align} 
where $\beta = \sqrt{\frac{1}{m} \left( \kappa + \frac{1}{m} \right) }$. 
\label{thm:coherenceupperboundmodd} 
\end{theorem} 

\begin{proof} 
This is proved in Appendix \ref{sec:generalkappa}. 
\end{proof}

It turns out that we can realize our method of constructing frames from cyclic groups in a much more general context, which we will detail in the upcoming sections. While Theorems \ref{thm:coherenceupperbound} and \ref{thm:coherenceupperboundmodd} will only generalize in certain scenarios, we will always be able to use Lemma \ref{lem:tightframecohbound} to bound our frames' coherence.

\section{Tight Group Frames and the Group Fourier Matrix} 
\label{sec:groupfouriermat} 

In light of Lemmas \ref{lem:tightframevar} and \ref{lem:tightframecohbound}, we will first establish the tools we need to ensure that our group frames are tight. It turns out that the tight group frames have been completely classified \cite{Vale, CasazzaKutyniok}. On this note, we review some basics on representation theory, which the interested reader can read about in greater depth in the first few chapters of \cite{Serre}. 

Let $G$ be a group of size $n$, and recall that a complex representation of $G$ is formally defined as a complex vector space $V$ together with a function $\rho: G \to GL(V)$ such that $\rho(g g') = \rho(g) \rho(g')$, $\forall g, g' \in G$. If $V$ has dimension $d$, then $\rho(g)$ is simply a $d \times d$ invertible complex matrix---a \textit{degree} $d$ representation. Two representations $\rho_1$ and $\rho_2$ with corresponding vector spaces $V_1$ and $V_2$ are \textit{equivalent} if there is an invertible transformation $T : V_1 \to V_2$ such that $T \rho_1(g) T^{-1} = \rho_2(g) $ for all $g \in G$. A basic result in representation theory says that every representation of a finite group is equivalent to a unitary representation, in which all the $\rho(g_i)$ are unitary matrices, which is why we have used the notation $\rho(g_i) = \U_i$ in our previous discussion. We will typically assume our representations are unitary without loss of generality. 

A representation $\rho$ is \textit{reducible} if there is a nontrivial subspace $V'$ of $V$ which is mapped to itself by $\rho(g)$ for every $g \in G$. Otherwise, it is called \textit{irreducible}. As matrices, the representation is reducible if the $\rho(g)$ can be simultaneously block-diagonalized by a similarity transformation. For any finite group $G$ of size $n$, there are only a finite number of inequivalent, irreducible unitary representations. If we call them $\rho_1$, ..., $\rho_{n_r}$ with corresponding degrees $d_1, ..., d_{n_r}$, then it can be shown \cite{Serre} that these degrees satisfy the relation 
\begin{align} 
\sum_{i = 1}^{n_r} d_i^2 = |G|. 
\label{eqn:irreddimsum} 
\end{align} 

Every complex representation of $G$ is equivalent to an orthogonal direct sum of irreducible representations. Formally, this means that there is an invertible linear transformation $T: V \to V_1 \oplus ... \oplus V_m$ such that the $V_i$ are mutually orthogonal vector spaces and $T \rho(g) T^{-1} = \rho_1(g) \oplus ... \oplus \rho_m(g)$, where for each $i$, $\rho_i$ and $V_i$ give an irreducible representation of $G$. These irreducible representations can again be taken to be unitary. As matrices, this means that the $\rho(g)$ can be simultaneously block-diagonalized in the form $\rho(g) = \text{diag}(\rho_1(g), ..., \rho_m(g))$. A basic result of representation theory is that this decomposition into irreducible components is unique up to isomorphism. We are now ready to give a classification of all the tight $G$-frames: 

\begin{theorem}[\cite{CasazzaKutyniok}] 
\label{thm:tightGframes} 

Let $G = \{g_i\}_{i = 1}^n$ be a finite group, and $\rho: G \to GL(V)$ a complex representation of $G$ which has the decomposition into orthogonal unitary irreducible representations: $$V = V_1 \oplus ... \oplus V_m, $$ $$\rho(g) = \rho_1(g) \oplus ... \oplus \rho_m(g). $$ Let $\vv = \vv_1 + ... + \vv_m$, $\vv_k \in V_k$, and set $\vect{f}_i = \rho(g_i) \vv$. Then $\{\vect{f}_i\}_{i = 1}^n$ is a tight $G$-frame if and only if 
\begin{itemize} 
\item $\frac{||\vv_i||_2^2}{||\vv_j||_2^2} = \frac{\text{dim}(V_i)}{\text{dim}(V_j)}$, and 
\item if the $i^{th}$ and $j^{th}$ irreducible components are equivalent via $T: V_i \to V_j$, then $T \vv_i$ and $\vv_j$ are orthogonal. 
\end{itemize} 

\end{theorem} 

\begin{proof} 
This is Theorem 5.4 in \cite{CasazzaKutyniok}. It follows from considering the frame matrix $\M := \begin{bmatrix} \hdots & \rho(g_i) \vv & \hdots \end{bmatrix}_{i = 1}^n$ and applying Schur's Lemma (Section 2.2, \cite{Serre}) to the product $\M\M^*$ to see when it is a scalar matrix, which is equivalent to the columns of $\M$ forming a tight frame. 
\end{proof} 

We will now establish a tool that will allow us to easily use this theorem to construct tight frames. On this note, consider the following well-studied generalization of the classical discrete Fourier transform \cite{Terras}: 

\begin{definition} 
We define the \textit{group Fourier transform} of a complex-valued function on $G$, $f: G \to \C$, to be the function that maps a degree $d$ representation $\rho$ to the $d \times d$ complex matrix 
\begin{align} 
\hat{f}(\rho) = \sum_{g \in G} f(g) \rho(g). 
\end{align} 
There is an inverse transformation given by 
\begin{align} 
f(g) = \frac{1}{|G|} \sum_{i = 1}^{n_r} d_i \Tr{\rho_i(g^{-1}) \hat{f}(\rho_i)}, 
\end{align} 
where the sum is taken over all the inequivalent irreducible representations of $G$. 
\end{definition}

Much like the traditional discrete Fourier transform, this transformation has a matrix representation in the form 
\begin{align} 
\mathcal{F} & = \begin{bmatrix} \sqrt{d_1} \text{vec}(\rho_1(g_1)) & \hdots & \sqrt{d_1} \text{vec}(\rho_1(g_n))\\ 
\sqrt{d_2} \text{vec}(\rho_2(g_1)) & \hdots & \sqrt{d_2} \text{vec}(\rho_2(g_n))\\ 
\vdots &  \ddots & \vdots \\ 
\sqrt{d_{n_r}} \text{vec}(\rho_{n_r} (g_1)) & \hdots & \sqrt{d_{n_r}} \text{vec}(\rho_{n_r} (g_n)) \end{bmatrix}, 
\label{eqn:matrixF} 
\end{align} 
where for a $d \times d$ matrix $A$, $\text{vex}(A)$ is the \textit{vectorization} of $A$, i.e., the vector formed by stacking the columns of $A$ into a single $d^2 \times 1$ column. From equation (\ref{eqn:irreddimsum}), we see that $\mathcal{F}$ is a square matrix. 

Notice that when $G$ is a cyclic group of size $n$, then the group elements are $\{0, 1, ..., n-1\}$ (with the group operation being addition modulo $n$). Since this group is abelian, there are exactly $n$ irreducible representations, $\{\rho_\ell\}_{\ell = 1}^n$, each degree-1. $\rho_\ell$ is simply the function that maps $k \mapsto \omega^{k \ell}$, $k \in \{0, ..., n-1\}$, where $\omega = e^\frac{2 \pi i}{n}$. In this case our group Fourier transform and matrix become the familiar discrete time Fourier transform and DFT matrix.

\begin{theorem} 
Let $G = \{g_i\}_{i = 1}^n$ be a finite group with inequivalent, irreducible representations $\{\rho_{i}\}_{i = 1}^{n_r}$, and $\mathcal{F}$ the group Fourier matrix of $G$ as in (\ref{eqn:matrixF}). Then the columns of $\mathcal{F}$ form a tight $G$-frame, so $\mathcal{F}$ is a unitary matrix. In fact, if $\tilde{\rho}: G \to \C^{d \times d}$ is a representation of $G$ and $\tilde{\vv} \in \C^{d \times 1}$ such that the columns of $\M := [\hdots ~ \tilde{\rho}(g_i) \tilde{\vv} ~ \hdots]_{i = 1}^n$ form a tight frame, then the rows of $\M$ are a subset of the rows of $\mathcal{F}$ up to an equivalence of $\tilde{\rho}$ or a change of basis of $\C^{d \times 1}$. 
\label{thm:groupfouriertightframes} 
\end{theorem} 

\begin{proof} 
The group Fourier matrix $\mathcal{F}$ can be realized as a $G$-frame as follows: For each $i = 1, ..., n_r$, define the representation 
\begin{align} 
\tilde{\rho}_i(g) = \text{diag}(\rho_i(g), ..., \rho_i(g)) \in \C^{d_i^2 \times d_i^2}, 
\label{eqn:tilderhoidef} 
\end{align} 
a direct sum of $d_i$ copies of the irreducible representation $\rho_i$. Also define the vector 
\begin{align} 
\vv_i = \text{vec}(\matr{I}_{d_i}) = [\vect{e}_i^{(1) T }, ..., \vect{e}_i^{(d_i) T }]^T \in \C^{d_i^2}, 
\label{eqn:videf} 
\end{align} 
where $\matr{I}_{d_i}$ is the $d_i \times d_i$ identity matrix and $\vect{e}_i^{(j)} \in \C^{d_i \times 1}$ is the $j^{th}$ column of $\matr{I}_{d_i}$---a vector of all zeros except for a 1 in the $j^{th}$ position. 

Now choose the representation 
\begin{align} 
\rho(g) = \text{diag}(\tilde{\rho}_1(g), \tilde{\rho}_2(g), ..., \tilde{\rho}_{n_r}(g)) \in \C^{n \times n}, 
\label{eqn:directprodrho} 
\end{align}
and the vector 
\begin{align} 
\vv = [\sqrt{d_1} \vv_1^T, \sqrt{d_2} \vv_2^T, \hdots, \sqrt{d_{n_r}} \vv_{n_r}^T]^T \in \C^{n}. 
\label{eqn:directprodv} 
\end{align} 
Then $\mathcal{F}$ is the $G$-frame with columns $\rho(g_i) \vv$. For any $i$, $\{\vect{e}_i^{(j)}\}_{j = 1}^{d_i}$ is a complete orthonormal set in $\C^{d_i}$, and $\frac{||e_{i_1}^{(j_1)}||_2^2}{||e_{i_2}^{(j_2)}||_2^2} = \frac{d_{i_1}}{d_{i_2}}$. From Theorem \ref{thm:tightGframes}, we see that not only do the columns of $\mathcal{F}$ form a tight $G$-frame, but in fact up to a change of basis of the $\vect{e}_i^{(j)}$ or a similarity transformation of the $\rho_i$, every tight $G$-frame can be realized as a subset of the rows of $\mathcal{F}$ by forming each $\vv_i$ from a corresponding \textit{subset} of the columns $\{\vect{e}_i^{(j)}\}_{j = 1}^{d_i}$. 
\end{proof}

Theorem \ref{thm:groupfouriertightframes} reduces the task of constructing tight $G$-frames to selecting blocks of rows of the corresponding group Fourier matrix $\mathcal{F}$. Our job will now be to find good choices of the group $G$, and to identify which rows of $\mathcal{F}$ to choose to create a tight group frame with low coherence. We should mention that this problem was explored for abelian groups $G$ in \cite{DingFeng}, with a focus on finding frames with coherence equal to the Welch Bound. We will find, however, that by not placing any restrictions on our group $G$, and by allowing our coherence to be slightly above the Welch lower bound, we can produce a vastly larger and richer collection of frames.

\section{Reducing the Number of Distinct Inner Products in Tight Group Frames} 
\label{sec:reducinginnerproducts} 

In our original construction from Theorem \ref{thm:cosetinnerproducts}, we designed harmonic frames in the form of $\M$ from (\ref{eqn:Meqn}) which arose from representations of the cyclic group $G = \Z/n\Z$, where $n$ is a prime. Indeed, the $j^{th}$ row of $\M$ is $[1, \omega^{k_j}, \omega^{2 k_j}, ..., \omega^{(n-1) k_j}$, where $\omega = e^{\frac{2 \pi i}{n}}$, and we can now see that this is simply the row of the group Fourier matrix of $G$ corresponding to the n-dimensional representation $\rho_{k_j}(\ell) = \omega^{\ell k_j}$, for $\ell \in \{0, ..., n-1\}$. We wish to generalize our original method from Theorem \ref{thm:cosetinnerproducts} of constructing frames with few distinct inner product values.

On this note, we will consider constructing frames by choosing the blocks of rows corresponding to $m$ of the representations, which we may assume are $\rho_1, ..., \rho_m$ up to a reordering, so that our frame matrix takes the form 
\begin{align} 
\M & = \begin{bmatrix} \sqrt{d_1} \text{vec}(\rho_1(g_1)) & \hdots & \sqrt{d_1} \text{vec}(\rho_1(g_n))\\ 
\vdots &  \ddots & \vdots \\ 
\sqrt{d_{m}} \text{vec}(\rho_{m} (g_1)) & \hdots & \sqrt{d_{m}} \text{vec}(\rho_{m} (g_n)) \end{bmatrix}. 
\label{eqn:Fmatmrows} 
\end{align} 
As an analog to Equations (\ref{eqn:directprodrho}) and (\ref{eqn:directprodv}) from the proof of Theorem \ref{thm:groupfouriertightframes}, this corresponds to the tight group frame whose elements are the images of the vector $\vv = [\sqrt{d_1} \vv_1^T, \sqrt{d_2} \vv_2^T, \hdots, \sqrt{d_{m}} \vv_{m}^T]^T $ under the representation $\rho(g) = \text{diag}(\tilde{\rho}_1(g), \tilde{\rho}_2(g), ..., \tilde{\rho}_{m}(g))$, where $\tilde{\rho}_i$ and $\vv_i$ are defined as in Equations (\ref{eqn:tilderhoidef}) and (\ref{eqn:videf}) respectively. The dimension of this representation is easily seen to be $\sum_{i = 1}^m d_i^2$. Note that in the setting of Theorem \ref{thm:cosetinnerproducts}, the representations $\rho_i$ are all 1-dimensional, so the block $\begin{bmatrix} \sqrt{d_i} \text{vec}(\rho_i(g_1)) & \hdots & \sqrt{d_i} \text{vec}(\rho_i(g_n)) \end{bmatrix}$ is just a single row. 

The inner product between the $i^{th}$ and $j^{th}$ columns of $\M$ in (\ref{eqn:Fmatmrows}) takes the form 
\begin{align} 
\hspace{-10pt} \sum_{t = 1}^m d_t \text{vec}(\rho_t(g_i))^* \text{vec}(\rho_t(g_j)) &= \sum_{t = 1}^m d_t \Tr{\rho_t(g_i)^*\rho_t(g_j)} \\
&= \sum_{t = 1}^m d_t \Tr{\rho_t(g_i^{-1} g_j)} \label{eqn:groupFTinnerprodstep}\\ 
&= \sum_{t = 1}^m d_t \chi_t(g_i^{-1} g_j). \label{eqn:groupFTinnerprodstep2}
\end{align} 
Here, $\chi_i(g) := \Tr{\rho_i(g)}$ is the \textit{character function} associated to the representation $\rho_i$. Equation (\ref{eqn:groupFTinnerprodstep2}) actually arises in \cite{DingFeng}, though only 1-dimensional representations are considered, in which case each representation is essentially just its own character. Note that in this form the frame is unnormalized, but all of the columns have the same norm, which is given by the square root of the inner product associated to the identity element: 
\begin{align} 
||\rho(g) \vv||_2 = \sqrt{\sum_{t = 1}^m d_t \chi_t(1)} = \sqrt{\sum_{t = 1}^m d_t^2}, \label{eqn:GFMcolnorm} 
\end{align} 
where we have used the fact the character evaluated at $1$ is simply the dimension of the representation. Alternatively, we could have simply seen this to be the norm of $\vv$ by speculation. 

Basic representation theory tells us that a character $\chi$ completely determines its representation up to isomorphism, and as such the characters of many groups are well-studied. In light of this fact, we can often compute the coherence of frames in the form of (\ref{eqn:Fmatmrows}) for different choices of representations $\{\rho_i\}_{i = 1}^m$ without explicitly building the frame matrix $\M$, which can often be a tedious computation. From (\ref{eqn:groupFTinnerprodstep}) and (\ref{eqn:groupFTinnerprodstep2}) we see that the inner product depends only on the group element $g_k := g_i^{-1} g_j$, so a priori there are only $n-1$ possible nontrivial distinct inner product values, and each of these values arises the same number of times as the inner product between two columns. This was to be expected, since the columns of $\M$ form a group frame in light of Theorem \ref{thm:groupfouriertightframes}. If we could generalize our method for choosing rows of the classical Fourier matrix, however, we could hope to reduce this number even further. 

Toward this end, we consider the group of \textit{automorphisms} of $G$. An automorphism of $G$ is a bijective function $\sigma: G \to G$ which respects the group multiplication, i.e. $\sigma(g g') = \sigma(g) \sigma(g')$ for any $g, g' \in G$. The automorphisms of $G$ form a group under composition, denoted $Aut(G)$. An important subgroup of $Aut(G)$ is that of the \textit{inner automorphisms}, denoted $Inn(G)$. These are the automorphisms which arise from \textit{conjugation} by an element $h \in G$, which is the function $\sigma_h(g) = h g h^{-1}$. Two elements $g$ and $g'$ are said to be \textit{conjugate} if there is some $h \in G$ such that $g' = h g h^{-1}$, and the set of all elements conjugate to $g$ is called the \textit{conjugacy class} $\mathcal{C}_g$. We see that the relation $\{g \sim g' \iff g \text{ is conjugate to } g' \}$ is an equivalence relation on $G$, so $G$ can be partitioned into a disjoint union of its conjugacy classes. $Inn(G)$ is easily verified to be a normal subgroup of $Aut(G)$, and the quotient group $Aut(G)/Inn(G)$ is called the group of \textit{outer automorphisms}, denoted $Out(G)$. 

Any conjugation $\sigma_h \in Inn(G)$ fixes a representation's character function. Indeed, if $\rho$ is a representation of $G$ with associated character $\chi$, then 
\begin{align} 
\chi(\sigma_h(g)) = \chi(h g h^{-1}) = \Tr{\rho(h) \rho(g) \rho(h)^{-1}} = \Tr{\rho(g)} = \chi(g). 
\end{align} 
Thus, since the inner products between the columns of $\M$ in (\ref{eqn:Fmatmrows}) can be expressed as in (\ref{eqn:groupFTinnerprodstep2}) in terms of the characters of the irreducible representations of $G$ (i.e. a so-called \textit{character function} on the group elements), we see that there is really only one inner product value for each conjugacy class of $G$. Note that while this observation has the advantage of reducing the number of distinct inner product values to consider, we unfortunately cannot readily apply Lemma \ref{lem:tightframecohbound} to obtain a tighter coherence bound since these values no longer occur with the same multiplicity. Indeed, for each $g \in G$, the corresponding inner product value $\sum_{t = 1}^m d_t \chi_t(g)$ will arise once for each element in the conjugacy class $\mathcal{C}_g$, and the conjugacy classes need not have the same size. 

Since an automorphism essentially preserves the structure of the group $G$, it is no surprise that it also preserves the structure of its representations: 

\begin{lemma} 
$\rho(g)$ is an irreducible representation of the finite group $G$ if and only if $\rho(\sigma(g))$ is also an irreducible representation for any $\sigma \in Aut(G)$. Furthermore, $\rho(g)$ and $\rho(\sigma(g))$ have the same degrees. 
\label{lem:autrep} 
\end{lemma} 

\begin{proof} 
If $\rho: G \to GL(V)$ is a representation, then composing with the automorphism $\sigma: G \to G$ yields a function $\rho \circ \sigma: G \to GL(V)$ which respects the group multiplication: $\rho(\sigma(g g')) = \rho(\sigma(g) \sigma(g')) = \rho(\sigma(g)) \rho(\sigma(g'))$. Thus, $\rho(\sigma(g))$ is a well-defined representation which clearly has the same dimension as $\rho(g)$. Furthermore, since $\sigma$ is a bijection of $G$, the matrices $\{\rho(\sigma(g))~:~g \in G\}$ are simply a permutation of the matrices $\{\rho(g)~:~g \in G\}$, so the first representation is irreducible if and only if the second is. 
\end{proof} 

If $\rho$ is a representation with character $\chi$, and $\sigma \in Aut(G)$, we will use the notation $\rho_{\sigma}$ to indicate the representation 
\begin{align} 
\rho_{\sigma}(g) := \rho(\sigma(g)), \label{eqn:rhosigmanotation} 
\end{align} 
which is irreducible if $\rho$ is. $\rho_{\sigma}$ has corresponding character 
\begin{align} 
\chi_{\sigma}(g) := \chi(\sigma(g)). \label{eqn:chisigmanotation} 
\end{align} 
Under this notation, if $\bold{1} \in Aut(G)$ denotes the identity automorphism $\bold{1}(g) = g$, then $\rho_{\bold{1}}$ and $\chi_{\bold{1}}$ are simply $\rho$ and $\chi$ respectively. From Lemma \ref{lem:autrep}, we see that $Aut(G)$ has a group action on the irreducible representations and characters of $G$ given by 
\begin{align} 
\sigma' \cdot \rho_{\sigma} := \rho_{\sigma \sigma'}, \label{eqn:rhosigmaaction} \\
\sigma' \cdot \chi_{\sigma} := \chi_{\sigma \sigma'}. \label{eqn:chisigmaaction} 
\end{align}

Let us consider case in our original construction from Theorem \ref{thm:cosetinnerproducts} where $G$ was the (additive) cyclic group $\Z/n\Z = \{0, ..., n-1\}$. In this case, $Aut(G)$ is isomorphic to the (multiplicative) group of elements relatively prime to $n$, $(\Z/n\Z)^\times$. For each $\ell \in (\Z/n\Z)^\times$, the corresponding automorphism $\sigma_\ell \in Aut(G)$ is given by $\sigma_\ell(g) = \ell g$. When we required that $n$ be prime in Theorem \ref{thm:cosetinnerproducts}, we ensured that every nonzero element had a multiplicative inverse modulo $n$, so in this case $(\Z/n\Z)^\times$ is the set $\{1, ..., n-1\}$. 

Refer back to the structure of our harmonic frame from (\ref{eqn:Meqn}): 
\begin{align} 
\M =  \frac{1}{\sqrt{m}} \begin{bmatrix} 1 &  \omega^{a_{1}} & \omega^{a_{1}\cdot 2} & \hdots & \omega^{a_{1}\cdot (n-1)}\\ 
						1 & \omega^{a_{2}} & \omega^{a_{2} \cdot 2} & \hdots & \omega^{a_{2}\cdot (n-1)}\\ 
						\vdots & \vdots & \vdots & \ddots & \vdots \\ 
						1& \omega^{a_{m}} & \omega^{a_{m} \cdot 2} & \hdots & \omega^{a_{m} \cdot (n-1)} \end{bmatrix},  
\label{eqn:Meqn2} 
\end{align} 
where $\omega = e^{2 \pi i / n}$. As we have discussed, selecting the frequencies $\{a_1, ..., a_m\}$ is equivalent to choosing rows of the group Fourier matrix corresponding to $\Z/n\Z$, each of which corresponds to a degree-1 representation. By choosing the frequencies $\{a_1, ..., a_m\}$ in (\ref{eqn:Meqn2}) to be a subgroup of $(\Z/n\Z)^\times$ as in Theorem \ref{thm:cosetinnerproducts}, we can now see that we are actually choosing a subgroup of $Aut(G)$. Without loss of generality, let $a_1 = 1$ so that the first row of (\ref{eqn:Meqn2}) corresponds to the representation $\rho(g) = \omega^{g}$. Then the $i^{th}$ row corresponds to the representation $\rho_i(g) := \rho(\sigma_{a_i}(g)) = \omega^{a_i g}$. Thus, we have formed $\M$ by choosing the rows of the group Fourier matrix corresponding to a subset of representations of the form $\{\sigma_i \cdot \rho\}$, where the $\{\sigma_i\}$ form a subgroup of automorphisms. 

We wish to generalize this process to groups $G$ other than $\Z/n\Z$ by choosing an irreducible representation $\rho$ of $G$ and taking its image under a subgroup of automorphisms $\{\sigma_i\} \le Aut(G)$. Note that from Lemma \ref{lem:autrep}, the representations $\{\sigma_i \cdot \rho\}$ will all be irreducible, and hence correspond to easily-identified blocks of rows from the group Fourier matrix $\mathcal{F}$ in (\ref{eqn:matrixF}). It is not clear, however, whether these representations will be distinct. The question now becomes how to choose the subgroup of automorphisms?


\section{Choosing the Automorphism Subgroup} 
\label{sec:choosingautomorphisms} 

Let $H \le Aut(G)$ be a group of automorphisms of $G$, and fix an irreducible representation $\rho$ with character $\chi$. Define $K$ to be the subgroup of $H$ which fixes $\chi$: 
\begin{align} 
K = \{ \sigma \in H ~: ~ \chi(\sigma(g)) = \chi(g), ~ \forall g \in G\}. \label{eqn:Kdefeqn} 
\end{align} 
Immediately we see that $K$ contains every inner automorphism in $H$. Thus, it is effectively the group of \textit{outer} automorphisms which acts nontrivially on the representations. Now choose a subgroup $A \le H$ such that the group product $KA := \{ka~:~k \in K, ~ a \in A\}$ is a subgroup of $H$. This is equivalent to the group products $KA$ and $AK$ being equal as sets. We consider choosing the rows of the generalized Fourier matrix corresponding to the representations $\{\rho_a~:~ a \in A\}$, with notation as in (\ref{eqn:rhosigmanotation}). From Lemma \ref{lem:autrep}, all of these representations have the same degree $d$. Thus, if $A = \{a_1, ..., a_m\} \le Aut(G)$, then $\M$ takes the form 
\begin{align} 
\M & = \sqrt{d} \begin{bmatrix}  \text{vec}(\rho_{a_1}(g_1)) & \hdots &  \text{vec}(\rho_{a_1}(g_n))\\ 
\vdots &  \ddots & \vdots \\ 
 \text{vec}(\rho_{a_m} (g_1)) & \hdots &  \text{vec}(\rho_{a_m} (g_n)) \end{bmatrix}. 
\label{eqn:Aframe} 
\end{align} 
Notice that if $A$ and $K$ have nontrivial intersection, then some of the blocks of rows of $\M$ above may correspond to repeated or isomorphic representations. If this is the case our frame will no longer be tight. We can avoid this by assuming that $|K \cap A| = 1$, though we will typically not make use of this assumption in our following proofs.

Now let us examine the inner products between our frame elements. From (\ref{eqn:groupFTinnerprodstep2}), the inner product corresponding to the group element $g$ is 
\begin{align} 
d \sum_{a \in A} \chi_a(g) & = d \sum_{a \in A} \chi(a(g)). \label{eqn:daginnerprod} 
\end{align} 
Our aim is to generalize the concept from Theorem \ref{thm:cosetinnerproducts} of having one inner product per coset of a subgroup of $Aut(G)$.  We first establish the following preliminary lemma:

\begin{lemma} 
Let $A$ and $K$ be subgroups of a finite group $H$ such that the set product $KA$ is a group, and let $\{a_i\}_{i = 1}^{|A|/|A\cap K|}$ be a set of right coset representatives for $(A \cap K) \backslash A$. Then for each fixed $a_i$ and $k \in K$, there is a unique $a_{i'}$ and $k' \in K$ such that $a_{i'} k = k' a_i$, and a unique $a_{i''}$ and $k'' \in K$ such that $a_i k = k'' a_{i''}$. 
\label{lem:defineki} 
\end{lemma} 

\begin{proof} 
Since $KA$ is a group (by assumption) which obviously contains both $K$ and $A$, we can write $a_i k = \tilde{k} \tilde{a}$ for some $\tilde{k} \in K$ and $\tilde{a} \in A$. Then $\tilde{a}$ can further be written uniquely in the form $\tilde{k}_2 a_{i''}$ for some $\tilde{k}_2 \in A \cap K$ and $a_{i''}$ one of the right coset representatives of $A \cap K$ in $A$. Setting $k'' = \tilde{k} \tilde{k}_2$ gives us the second part of this theorem. 

Now suppose there are two pairs $(a_j, k'_j)$ and $(a_t, k'_t)$ such that 
\begin{align} 
a_j k = k'_j a_i, \label{eqn:akeq1} \\ 
a_t k = k'_t a_i. \label{eqn:akeq2} 
\end{align} 
Then from (\ref{eqn:akeq2}) we have $a_t (a_j)^{-1} a_j k = k'_t (k'_j)^{ -1} k'_j a_i$, and we can use (\ref{eqn:akeq1}) to cancel out $a_j k$ and $k'_j a_i$ from this expression to arrive at 
\begin{align} 
a_t (a_j)^{-1} = k'_t (k'_j)^{-1} \in A \cap K. 
\end{align} 
But since $a_t$ and $a_j$ are representatives of distinct right cosets of $A \cap K$ in $A$, they must be equal, hence $a_t = a_j$ and $k'_t = k'_j$. This shows that there can only be \textit{at most} one pair $(a_{i'}, k')$ such that $a_{i'} k = k' a_i$. But since we have already shown that every $a_j k$ can be written uniquely in the form $k'' a_{j''}$ for some $a_{j''}$, then since our groups are finite there must be some $j$ for which $a_{j''} = a_i$, so there is \textit{exactly} one such pair $(a_{i'}, k') = (a_j, k'')$ which satisfy the hypotheses of the lemma. 
\end{proof}

The next lemma now extends the coset idea of Theorem \ref{thm:cosetinnerproducts} to drastically reduce the number of distinct inner product values we need consider. 

\begin{lemma} 
Let $G$ be a finite group, $H \le Aut(G)$, $\rho$ an irreducible representation of $G$ with character $\chi$, and $K$ the subgroup of $H$ which fixes $\chi$ as in (\ref{eqn:Kdefeqn}). Let $A$ be a subgroup of $H$ such that $KA$ is a group. Then for any $\sigma_1, \sigma_2 \in H$ which are in the same \textit{right} coset of $KA$, the inner products associated to $\sigma_1(g)$ and $\sigma_2(g)$ respectively are equal for any $g \in G$. That is, 
\begin{align} 
d \sum_{a \in A} \chi_a(\sigma_1(g)) = d \sum_{a \in A} \chi_{a}(\sigma_2(g)). 
\end{align}  
\label{lem:gencosetinnerprod} 
\end{lemma}

\begin{proof} 
Since $\sigma_1$ and $\sigma_2$ are in the same right coset of $KA$ (which is equal to $AK$), there is some $h \in H$ such that $\sigma_1 = a_1 k_1 h$ and $\sigma_2 = a_2 k_2 h$ for some $a_1, a_2 \in A$ and some $k_1, k_2 \in K$. Thus, (\ref{eqn:daginnerprod}) becomes 
\begin{align} 
d \sum_{a \in A} \chi(a \sigma_1(g)) & = d \sum_{a \in A} \chi(a a_1 k_1h(g)) \\ 
							& = d \sum_{a \in A} \chi(a k_1h(g)), 
\end{align} 
which follows from the fact that multiplication by $a_1$ permutes the elements of $A$. 

Now let $\{a_i\}$ be a set of right coset representatives for $(A \cap K)\backslash A$. Our sum now becomes 
\begin{align} 
d \sum_{a \in A} \chi(a k_1h(g)) & = d \sum_{a_i} \sum_{\gamma \in A \cap K} \chi(\gamma a_i k_1h(g)) \\ 
						& = d \sum_{a_i} |A \cap K| \chi(a_i k_1 h(g)), 
\end{align} 
which follows from fact that elements of $K$ fix $\chi$. Now for each $a_i$, we know from Lemma \ref{lem:defineki} that $a_i k_1$ is uniquely expressible in the form $k'_j a_j$ for some right coset representative $a_j$ and some $k'_j \in K$. Thus, since the $\{a_i\}$ and $\{a_j\}$ are in one to one correspondence by Lemma \ref{lem:defineki}, we can further rewrite our sum as 
\begin{align} 
d \sum_{a_i} |A \cap K| \chi(a_i k_1 h(g)) & = d \sum_{a_j} |A \cap K| \chi(k'_j a_j h(g)) \\ 
								& = d \sum_{a_j} |A \cap K| \chi(a_j h(g)) \\ 
								& = d \sum_{a_j} \sum_{\gamma \in A \cap K} \chi(\gamma a_j h(g))\\ 
								& = d \sum_{a \in A} \chi(a h(g)). 
\end{align} 
Since the inner product depends only on $h$, we are done. 
\end{proof} 

We can now express each inner product in terms of a right coset of $KA$ and an \textit{orbit} of $G$ under the automorphism group $H$. Two elements $g, g' \in G$ are said to be in the same orbit if there is an automorphism $h \in H$ such that $h(g) = g'$. Note that since $g = h^{-1}(g')$, this is an equivalence relation, so the orbits partition $G$. We may write this orbit as $Hg := \{ h(g)~|~h \in H\}$, and we say that $g$ is a \textit{representative} of this orbit. It should be clear that the identity element $1 \in G$ is in its own orbit.


We are now equipped to bound both the number of distinct inner product values, as well as the coherence of our new frames. The following theorem contains the analogs of Lemma \ref{lem:tightframecohbound} and Theorem \ref{thm:cosetinnerproducts} to the broader class of frames we have just constructed. 

\begin{theorem} 
Let $G$ be a finite group of size $n$ and $\rho$ a degree-$d$ irreducible representation of $G$ with character $\chi$. Define
\begin{itemize} 
\item $H \le Aut(G)$ a group of automorphisms of $G$, 
\item $K := \{ \sigma \in H ~: ~ \chi(\sigma(g)) = \chi(g), ~ \forall g \in G\}$, the subgroup of $H$ consisting of automorphisms which fix $\chi$, 
\item $A = \{a_i\}_{i = 1}^m \le H$, any subgroup of $H$ such that the set product $KA$ is also subgroup of $H$ with $A \cap K = 1$, 
\item $\{h_i\}_{i = 1}^{n_c}$ representatives of the distinct cosets of $KA$ in $H$ 
\item $\{g_j\}_{j = 1}^{n_o}$ representatives of the distinct orbits of $G$ under $H$ 
\end{itemize} 
Finally, let $\M$ be the frame with elements $\{\sqrt{d} [\text{vec}(\rho_{a_1}(g))^T, ..., \text{vec}(\rho_{a_m}(g))^T]^T\}_{g \in G}$ as in (\ref{eqn:Aframe}). Then $\M$ is a tight frame with at most $n_c(n_o - 1)$ distinct inner product values between its vectors. If $\mu_W$ is the lower bound on coherence given by the Welch bound (explicitly $\mu_W = \sqrt{\frac{n - dm}{dm(n - 1)}})$, then the coherence $\mu$ of our frame is bounded by 
\begin{align} 
\mu \le \sqrt{\frac{|G| - 1}{\min_{\{(i, j) : g_j \ne 1\}} |KAh_i(g_j)|}} \mu_W. 
\label{eqn:generalizedfourierframebound} 
\end{align}  
\label{thm:generalizedfourierframebound} 
\end{theorem} 

\begin{proof} 
By hypothesis, $G$ is partitioned into distinct orbits $Hg_1, ..., Hg_{n_o}$ with representatives $g_1, ..., g_{n_o}$. Let $g \in G$ be in the $j^{th}$ orbit so that for some $h \in H$ we have $h(g_j) = g$. Suppose that $h \in KA h_i$. Then from Lemma \ref{lem:gencosetinnerprod}, the inner product associated to $g$ is 
\begin{align} 
d \sum_{a \in A} \chi(a (g)) & = d \sum_{a \in A} \chi(a h(g_j)) = d \sum_{a \in A} \chi(a h_i (g_j)). 
\end{align} 
Thus, excluding the orbit corresponding to the identity element (which corresponds to taking the inner product of a column of $\M$ with itself), the number of nontrivial inner products that we must consider is $n_c ( n_o - 1)$, and the number of times the inner product corresponding to the pair $(h_i, g_j)$ arises is 
\begin{align} 
|KAh_i(g_j)| = \#\{kah_i(g_j)~:~  k \in K, ~ a \in A\}. 
\end{align} 

Now since our frame $\M$ is tight by Theorem \ref{thm:groupfouriertightframes}, then from Lemma \ref{lem:tightframecohbound}, the mean squared inner product between the frame vectors is equal to $\mu^2_W$, and this mean can be written as 
\begin{align} 
\mu_W^2 & = \frac{1}{\sum_{h_i} \sum_{g_j \ne 1}|KAh_i(g_j)|} \cdot \sum_{h_i} \sum_{g_j \ne 1} |KAh_i(g_j)| |\alpha_{i, j}|^2 \label{eqn:sumhgalphaismu1}\\ 
& = \frac{1}{|G| - 1} \cdot \sum_{h_i} \sum_{g_j \ne 1}  |KAh_i(g_j)| |\alpha_{i, j}|^2, \label{eqn:sumhgalphaismu2} 
\end{align} 
where $\alpha_{i, j}$ is the inner product associated to the pair $(h_i, g_j)$. From this, it follows that $$(|G| - 1) \mu^2_W \ge \left(\min_{\{(i, j) : g_j \ne 1\}}|KAh_i(g_j)|\right) \left( \max_{\{(i, j) : g_j \ne 1\}} |\alpha_{i, j}|^2\right), $$ from which our result follows. 

\end{proof}

We can see from Theorem \ref{thm:generalizedfourierframebound} that in general our coherence will be closer to the Welch bound if we have fewer orbits, and the sets $KAh_i(g_j)$ are close to each other in size. We articulate this in the following corollary. 

\begin{corollary} 
In Theorem \ref{thm:generalizedfourierframebound}, if the sets $KAh_i(g_j)$ are the same size for all $h_i$ and all nonidentity $g_j$, we achieve our optimal upper bound in (\ref{eqn:generalizedfourierframebound}): 
\begin{align} 
\mu & \le \sqrt{n_c(n_o - 1)} \mu_W.  
\label{eqn:ncnomuw} 
\end{align} 
\label{cor:samesizeorbits} 
\end{corollary} 

\begin{proof} 
If there are $n_c$ cosets of $KA$ in $H$, and $n_o$ orbits of $G$ under the action of $H$, then since $\sum_{h_i} \sum_{g_j \ne 1} |KAh_i(g_j)| = |G| - 1$, we have 
\begin{align} 
\min_{\{(i, j):g_j \ne 1\}} |KAh_i(g_j)| \le \frac{|G| - 1}{n_c(n_o - 1)}, 
\end{align} 
with equality if and only if the sets $KAh_i(g_j)$ are all the same size. The result follows immediately. 
\end{proof} 


For clarity, let us reiterate how our frames from Theorem \ref{thm:cosetinnerproducts} fall into the more general framework of Theorem \ref{thm:generalizedfourierframebound}. In this case,  

\begin{itemize} 
\item $G$ is the cyclic additive group $\Z/n\Z = \{0, 1, ..., n-1\} \mod n$, where $n$ is a prime. 
\item $\rho$ is the representation $\rho(x) = e^\frac{2 \pi i x}{n}$ for any $x \in G$. 
\item $\chi(x)$ is equal to $\rho(x)$ for any $x \in G$, since $\rho$ is a degree-1 representation. 
\item $H$ is the multiplicative group $(\Z/n\Z)^\times = \{1, 2, ..., n-1\} \mod n$, where each element $h \in (\Z/n\Z)^\times$ is viewed as an automorphism $h(x) = h \cdot x$. 
\item $K$ is the subgroup of $H$ such that $e^\frac{2 \pi i k x}{n} = e^\frac{2 \pi i x}{n}$, $\forall x \in G$. In this case, we can see that $K = \{1\}$. 
\item $A$ is the size $m$ subgroup of $H$, where $m | (n-1)$. Since $K$ is trivial, $KA$ is automatically a subgroup of $H$, and $A \cap K = 1$. 
\item $n_c$ is the number of cosets of $A$ in $H$, which is $\frac{n-1}{m}$. $\{h_i\}_{i = 0}^{n_c}$ are the representatives of these cosets. If $x$ is a cyclic generator for $H$, then the $h_i$ can be taken to be the powers of $x$: $h_i = x^i$, $i = 1, ..., n_c$. 
\item $n_o = 2$, because there are only two orbits of $G$ under $H$. One of these is the trivial orbit, $\{0\}$, and indeed $h \cdot 0 = 0$, $\forall h \in H$. All the nonzero elements $\{1, ..., n-1\} \subset G$ are in the same orbit, since any two of these elements differ only by a multiplicative factor in $H$. Thus we may take our two orbit generators to be $g_1 = 1$ (the generator of the nontrivial orbit) and $g_2 = 0$ (the generator of the tribal orbit). 
\end{itemize} 

In light of this last point, we see that these frames trivially satisfy the hypothesis of Corollary \ref{cor:samesizeorbits} since the sets $KAh_i(g_j)$ are simply the cosets $Ah_i$, which all have the same size as desired. (Note that since we write $G$ additively in this situation, the identity element is $0$ instead of $1$, so the hypothesis of Corollary \ref{cor:samesizeorbits} effectively becomes that the sets $KAh_i(g_j)$ are the same size for $g_j \ne 0$). Thus the frames from Theorem \ref{thm:cosetinnerproducts} give us our optimal bound in Theorem \ref{thm:generalizedfourierframebound}, and the bound in (\ref{eqn:ncnomuw}) becomes $\mu \le \sqrt{n_c} \mu_W$, which is the same bound we saw in Corollary \ref{cor:samesizeorbits}. We will explore this connection more in the next section. 


\section{Subgroups and Quotients of General Linear Groups} 
\label{sec:matrixgroups} 
We will now identify a class of groups that yield frames with remarkably low coherence using this framework, a subclass of which consists of the groups used in Theorem \ref{thm:cosetinnerproducts}. Recall that in our original construction of Theorem \ref{thm:cosetinnerproducts}, we chose $G$ to be the additive group $\Z/n\Z$, where $n$ was a prime $p$, and $H$ was isomorphic to the multiplicative group $(\Z/n\Z)^\times$, which contains all the nonzero elements of $\Z/n\Z$ when $n$ is prime. This is equivalent to choosing $G$ and $H$ respectively to be the additive and multiplicative groups of the finite field with $p$ elements, $\F_p$. In this case, $H$ is the simplest example of a general linear group. Indeed, $H$ can be interpreted as the 1-dimensional invertible matrices with entries in $\F_p$. As we will now see, subgroups and quotients of matrix groups over finite fields lend themselves naturally to our construction.

\subsection{Frames from Vector Spaces Over Finite Fields} 
\label{sec:vectorspacesfinitefields} 

Recall from our discussion following Theorem \ref{thm:generalizedfourierframebound} that in general our coherence will be closer to the Welch bound if we have fewer orbits, and the sets $KAh_i(g_j)$ are close to each other in size. The optimal case is when their sizes are all equal, in which case we obtain the bound in Corollary \ref{cor:samesizeorbits}. 
Equation (\ref{eqn:ncnomuw}) in this corollary closely resembles the result from Lemma \ref{lem:tightframecohbound}. This is no coincidence, since the condition that the sets $KAh_i(g_j)$ have the same size is equivalent to requiring that each corresponding inner product value arises the same number of times as the inner product between two frame elements. (Recall that we exploited this latter property in deriving Lemma \ref{lem:tightframecohbound}.) In a sense, the best case is when we have exactly one nontrivial orbit, so that $n_o = 2$. And if in addition the sets $KAh_i(g_j)$ have the same size for all $h_i$ and $g_j \ne 1$, Corollary \ref{cor:samesizeorbits} shows that the coherence is bounded by a factor of $\sqrt{n_c}$ of the Welch bound.

We saw at the end of Section \ref{sec:choosingautomorphisms} that this happens in our original frames constructed in Theorem \ref{thm:cosetinnerproducts}, when $G$ was the additive group of a prime-sized finite field $\F_p \cong \Z/p\Z$ and $H$ the set of automorphisms given by multiplication by elements of $\F_p^\times \cong (\Z/p\Z)^\times$. As we remarked at the beginning of this section, $H$ is the simplest example of a general linear group $GL(r, \F_p)$---the multiplicative group of $r \times r$ invertible matrices with entries in $\F_n$ (in this case $r=1$). It turns out that even higher-dimensional general linear groups fit the framework of Corollary \ref{cor:samesizeorbits}. If we set $H := GL(r, \F_p)$ then it is the automorphism group of $G := (\F_p)^r$, the $r$-dimensional vector space over $\F_p$ (viewed only as an additive abelian group). For any two nonzero vectors $\vv_1$ and $\vv_2$ in $(\F_p)^r$, there is an invertible matrix $\matr{W} \in GL(r, \F_p)$ such that $\matr{W}\vv_1 = \vv_2$, so all nontrivial elements of $(\F_p)^r$ lie in the same orbit under $H$. 



Alternatively, we may view $(\F_p)^r$ as the additive group of the finite field with $p^r$ elements, $\F_{p^r}$, which is a vector space over its subfield $\F_p$. An irreducible representation $\rho$ of $\F_{p^r}$ (and hence of $(\F_p)^r$) is the function 
\begin{align} 
\rho(x) = e^{\frac{2 \pi i \text{Tr}(x)}{p}}, 
\label{eqn:rhoexptrace} 
\end{align}
where $\text{Tr}(x)$ is the \textit{trace} of the field element $x$, defined as 
\begin{align} 
\text{Tr}: &\F_{p^r} \to \F_p, \nonumber\\ 
\Tr{x} & = x + x^p + x^{p^2} + ... + x^{p^{r-1}}.  
\end{align}
The trace function in our context is the sum of the automorphisms of $\F_{p^r}$ fixing the subfield $\F_{p}$, and is so named because $\Tr{x}$ is the trace of the matrix associated with the linear transformation of multiplication by $x$. This transformation acts on the additive group of $\F_{p^r}$ viewed as a vector space over $\F_p$. As such, the trace is an additive function: $\Tr{x + y} = \Tr{x} + \Tr{y}$, and consequently $\Tr{-x} = - \Tr{x}$. In the case where $r = 1$, the trace becomes the identity function, and we see that as expected we recover a familiar representation of $\F_p$ similar to the ones used in Theorem \ref{thm:cosetinnerproducts}. 

We should point out that the \textit{general} form of an irreducible representation of the additive group of $\F_{p^r}$ is $\rho_a(x) := \omega^{\Tr{ax}}$, where $\omega = e^{\frac{2 \pi i}{p}}$ and $a \in \F_{p^r}$. This is the image of the function $\rho$ in (\ref{eqn:rhoexptrace}) under the action of $k$ viewed as a matrix in $GL(r, \F_p)$ as just described. As such, it is fitting that the notation ``$\rho_k$'' bears resemblance to that of equation (\ref{eqn:rhosigmanotation}). Note also that since each of these is a degree-1 representation, each is equal to its own character: $\chi_k(x) = \rho_k(x)$. Each of the $p^r$ representations $\rho_k$, $k \in \F_{p^r}$, is unique, and from equation (\ref{eqn:irreddimsum}) we see that they indeed comprise \textit{all} of the inequivalent irreducible representations of $G$. 

Now, we can concisely describe the group $K$ of character-preserving automorphisms from Theorem \ref{thm:generalizedfourierframebound} as follows: $K$ is simply the set of automorphisms in $H$ which preserve the field trace, $K = \{k \in H ~|~\Tr{kx} = \Tr{x}, ~\forall x \in G\}$. It can be easily shown that the size of $K$ is $|K| = |H|/|\F_{p^r}^\times| = (p^r - p) \hdots (p^r - p^{r-1})$. What is not clear, however, is the form that each element of $K$ will take as a matrix in $H = GL(r, \F_p)$. The same issue arises when we attempt to compute the group $A$ from Theorem \ref{thm:generalizedfourierframebound}. 

To rectify this issue, we will shift our focus to the interpretation of $G$ as the additive group of the field $\F_{p^r}$. And instead of choosing $H$ to be the entire automorphism group $GL(r, \F_p)$, we will let $H$ be the size-$(p^r-1)$ subgroup of matrices corresponding to the nonzero field elements $\F_{p^r}^\times$. (Recall, each element of $\F_{p^r}^\times$ acts linearly on $\F_{p^r}$ by multiplication, and as such has a matrix representation when viewed as a linear transformation of $(\F_p)^r$.) In this new setting, the only element of $H$ which fixes the field trace is $1$, so $K$ is now the trivial group.

It is reasonable to ask if we lose anything by choosing $H$ to be only a proper subgroup of $GL(r, \F_p)$. But in fact, we can see from Lemma \ref{lem:gencosetinnerprod} and Theorem \ref{thm:generalizedfourierframebound} that the coherence of our frames depends only on the right cosets of $K$ in $H$. The following lemma shows that we do not lose anything by choosing $H$ to be $\F_{p^r}^\times$ instead of $GL(r, \F_p)$: 

\begin{lemma} 
Let $G = (\F_p)^r$ (which is the additive group of $ \F_{p^r}$), and $\chi$ a character of $G$. Let $H_1 = GL(r, \F_p)$ with $K_1 \le H_1$ the subgroup that fixes $\chi$, and $H_2 = \F_{p^r}^\times \le H_1$ with corresponding subgroup $K_2 = H_2 \cap K_1$. For every subgroup $A_1$ of $H_1$ with $A_1 \cap K_1 = 1$, there is a subgroup $A_2$ of $H_2$ with $A_2 \cap K_2 = 1$ such that the groups $A_1$ and $A_2$ give rise to the same inner products described by Lemma \ref{lem:gencosetinnerprod}.  
\end{lemma} 

\begin{proof} 
As we touched on above, since our character is a function of the form $\chi(x) = e^\frac{2 \pi i \Tr{a x}}{p}$, we observe that no nontrivial element of $H_2$ fixes $\chi$. Thus, $K_2 = 1$. Since the right cosets $K_1 \ H_1 = \{K_1 h_1~:~h_1 \in H_1\}$ partition $H_1$, each element $h_2 \in H_2$ must lie in some such coset. We claim that no two elements of $H_2$ are in the same right coset of $K_1$. To see this, assume we have $h_2$ and $h_2'$ in $H_2$ which lie in the same right coset of $K_1$. This means that $h_2' h_2^{-1} \in K_1 \cap H_2 = K_2$, hence $h_2$ and $h_2'$ must be equal. Furthermore, we know that there is one element of $H_2$ in \textit{each} right coset of $K_1$ in $H_1$, since \textit{every} character of $G$ can be written in the form $\chi(h_2(x))$ for some multiplicative field element $h_2 \in H_2$. (This is a well-known fact that can be found, for example, in \cite{Reeder}.) 

Now, if $A_1$ is a subgroup of $H_1$ which intersects $K_1$ trivially, each element of $A_1$ must lie in a distinct right coset of $K_1$. For each element $a_1 \in A_1$, let $a_2$ be the unique element of $H_2$ which lies in the same such coset, and let $A_2$ be the set of all these elements. Clearly $A_2$ has trivial intersection with $K_2$, since it is a subset of $H_2$. The fact that $A_2$ is itself a group is easy to verify. For example, for elements $a_2$ and $a_2'$ in $A_2$, with corresponding elements $a_1$ and $a_1'$ in $A_1$, we see that the product $a_2' a_2^{-1}$ is also an element of $A_2$ since it is the field element lying in the same right coset of $K_1$ as $a_1' a_1^{-1} \in A_1$. Since elements $a$ in the same right coset of $K_1$ give rise to the same character $\chi(a(x))$, we see also that the groups $A_1$ and $A_2$ will give rise to the same frame inner products as described in Lemma \ref{lem:gencosetinnerprod}. 
\end{proof}






Let us explicitly match this example with the framework of Theorem \ref{thm:generalizedfourierframebound}. We note that 
\begin{itemize} 
\item $G$ is the additive group of the vector space $(\F_p)^r$, or equivalently the additive group of the field $\F_{p^r}$. 

\item $\rho(x) = e^{\frac{2 \pi i \text{Tr}(x)}{p}}$.  

\item $\chi(x) = \rho(x)$, since $\rho$ is a 1-dimensional representation, hence is equal to its own character.  

\item $H = \F_{p^r}^\times = \F_{p^r} \setminus \{0\}$, where $G$ is viewed as the additive group of $\F_{p^r}$. Basic field theory tells us that $H$ is isomorphic to the cyclic group of size $p^r - 1$. 

\item $K = 1$, since the only field element $h \in H$ such that $\chi(h(x)) = \chi(x)$ is the identity. 

\item $A = \{a_1, ..., a_m\}$ is any subgroup of $H$, which will necessarily be a cyclic group of size $m$, where $m$ is a divisor of $p^r - 1$. Since $H$ is cyclic, there is a unique subgroup for each such $m$, and it consists of the $\left( \frac{p^r - 1}{m}\right)^{th}$ powers in $H$. Thus, if $x$ is a cyclic generator for $H$, we may set $y = x^{\frac{p^r - 1}{m}}$ and $a_i = y^i$ for each $i = 1, ..., m$. 

\item $n_c = \frac{p^r - 1}{m}$, the number of cosets of $A$ in $H$. If $x$ is a generator for the cyclic group $H$, these cosets are $h_i = x^i$, $i = 1, ..., n_c$. 

\item $n_o = 2$, since again $0 \in \F_{p^r}$ is in its own orbit, and all the nontrivial elements are in their own orbit under $H$ (generated by $1 \in \F_{p^r}$). 

\end{itemize}

Our new frame matrix $\M$ from (\ref{eqn:Aframe}) becomes 
\begin{align} \M  = \begin{bmatrix} \omega^{\Tr{a_1x_1}} & \omega^{\Tr{a_1x_2}} & \hdots & \omega^{\Tr{a_1 x_{n}}} \\ 
						    \vdots & \vdots & \vdots & \ddots & \vdots \\ 
						    \omega^{\Tr{a_m x_1}} & \omega^{\Tr{a_m x_2}} & \hdots & \omega^{\Tr{a_m x_{n}}}\end{bmatrix}, 
\label{eqn:Mfieldmat} 
\end{align} 
where we have expressed the elements of our field as $\{x_i\}_{i = 1}^{n}$. If $x_j - x_i = x_\ell$, the inner product between the $i^{th}$ and $j^{th}$ columns now becomes 
\begin{align} 
\sum_{a_t} \left(\omega^{\Tr{a_t x_i}} \right)^* \left(\omega^{\Tr{a_t x_j} } \right) & = \sum_{a_t} \omega^{\Tr{a_t (x_j - x_i)}} \\
& = \sum_{a_t} \omega^{\Tr{a_t x_\ell} }. 
\label{eqn:fieldinnerprod} 
\end{align} 

We can see from (\ref{eqn:fieldinnerprod}) that as in our original frames from Theorem \ref{thm:cosetinnerproducts}, we have exactly $\frac{n-1}{m}$ nontrivial inner product values: one for each element of $\F^\times_{p^r}$ (each of which represents a right coset of $K$ in $H$). Again, each of these values arises as an inner product the same number of times.

Since these new frames are a generalization our original frames constructed in Theorem \ref{thm:cosetinnerproducts}, it should come as no surprise that the bounds in Theorems \ref{thm:coherenceupperbound} and \ref{thm:coherenceupperboundmodd} generalizes as well: 

\begin{theorem} 
\label{thm:groupframebound} 
If $n$ is prime power $p^r$, $m$ a divisor of $n-1$, and $\{a_i\}$ the unique subgroup of $\F^\times_{p^r}$ of size $m$, then setting $\omega = e^{\frac{2 \pi i}{p}}$, and $\kappa := \frac{n-1}{m}$, the coherence $\mu$ of our frame $\M$ in (\ref{eqn:Mfieldmat}) satisfies 
\begin{align} 
 \mu & \le \frac{1}{\kappa} \left( (\kappa-1) \sqrt{\frac{1}{m} \left(\kappa + \frac{1}{m}\right)} + \frac{1}{m} \right). 
\end{align} 
\noindent 
If both $p$ and $m$ are odd, $\mu$ satisfies the tighter bound 
\begin{align} 
 \mu & \le \frac{1}{\kappa} \sqrt{\left(\frac{1}{m} + \left(\frac{\kappa}{2} - 1 \right) \beta \right)^2 + \left(\frac{\kappa}{2}\right)^2 \beta^2},  
\end{align} 
where $\beta = \sqrt{\frac{1}{m} \left( \kappa + \frac{1}{m} \right) }$. 

\end{theorem} 

\begin{proof} 
We present this proof in Appendix \ref{sec:generalkappa}. 
\end{proof}

\subsection{Smaller Alphabets and Frames from Hadamard Matrices} 
\label{sec:HadamardFrames} 

We emphasize that these generalized frames have several advantages over the original frames constructed in Theorem \ref{thm:cosetinnerproducts}. First, the number $n$ of frame vectors is no longer limited to being a prime, but is instead a power of a prime, $n = p^r$. Furthermore, the entries of our frame matrix $\M$ in  (\ref{eqn:Mfieldmat}) are no longer $n^{th}$ roots of unity, but rather $p^{th}$ roots of unity. This allows for more practical implementations of our frames. Indeed, while our original frames did achieve low coherence, the entries of the frame vectors came from an alphabet size as large as the frame itself. Thus even for small examples our frames could require an alphabet size of at least several hundred. In our new frames, we could fix $p$ to be a small prime and take a number of frame elements that is substantially larger, yet our frame vectors will only have entries from an alphabet of size $p$.

For instance, if $p = 2$, then even though our frame can have $n = 2^r$ elements for any $r$, the matrix $\M$ will always have $\pm 1$ entries. In this case, we have the following: 

\begin{theorem} 
When $p = 2$ in our above framework, our frame matrix $\M$ in (\ref{eqn:Mfieldmat}) is a subset of rows of an $n \times n$ Hadamard matrix. 
\end{theorem} 

\begin{proof} 
We already commented above that when $p = 2$, $\M$ will have $\pm 1$ entries. The theorem then follows from the fact that the frame is tight (i.e. the rows of $\M$ are orthogonal with equal norm) by Theorem \ref{thm:groupfouriertightframes}. 
\end{proof}

This is not the first time that frames with $\pm 1$ entries have been explored. For example, \cite{MixonBajwaCalderbank} designed such frames using codes constructed by \cite{Berlekamp1970} and \cite{YuGong}, and analyzed the frames' geometry. Figure \ref{fig:Hadinnprodhist_341_1024} illustrates the benefit of using our frames to control coherence. Depicting histograms of the inner products resulting from selecting two sets of 341 rows of from a $1024 \times 1024$ Hadamard matrix using our method (red) versus randomly (blue), we can see that our construction actually yields just two distinct inner product values in this case, both much closer to zero than the largest magnitude inner products from the random case. In Table \ref{table:Hadamardcoherencelist}, we compute the coherences of several random vs. group Hadamard frames, and compare to the Welch bound for reference. The group Hadamard frames have consistently lower coherence than the random Hadamard frames, particularly when the frame dimensions $m$ and $n$ are large but the quotient $\kappa = \frac{n-1}{m}$ is small.

\begin{figure}
\begin{center}
 \includegraphics[height = 200pt, width=300pt]{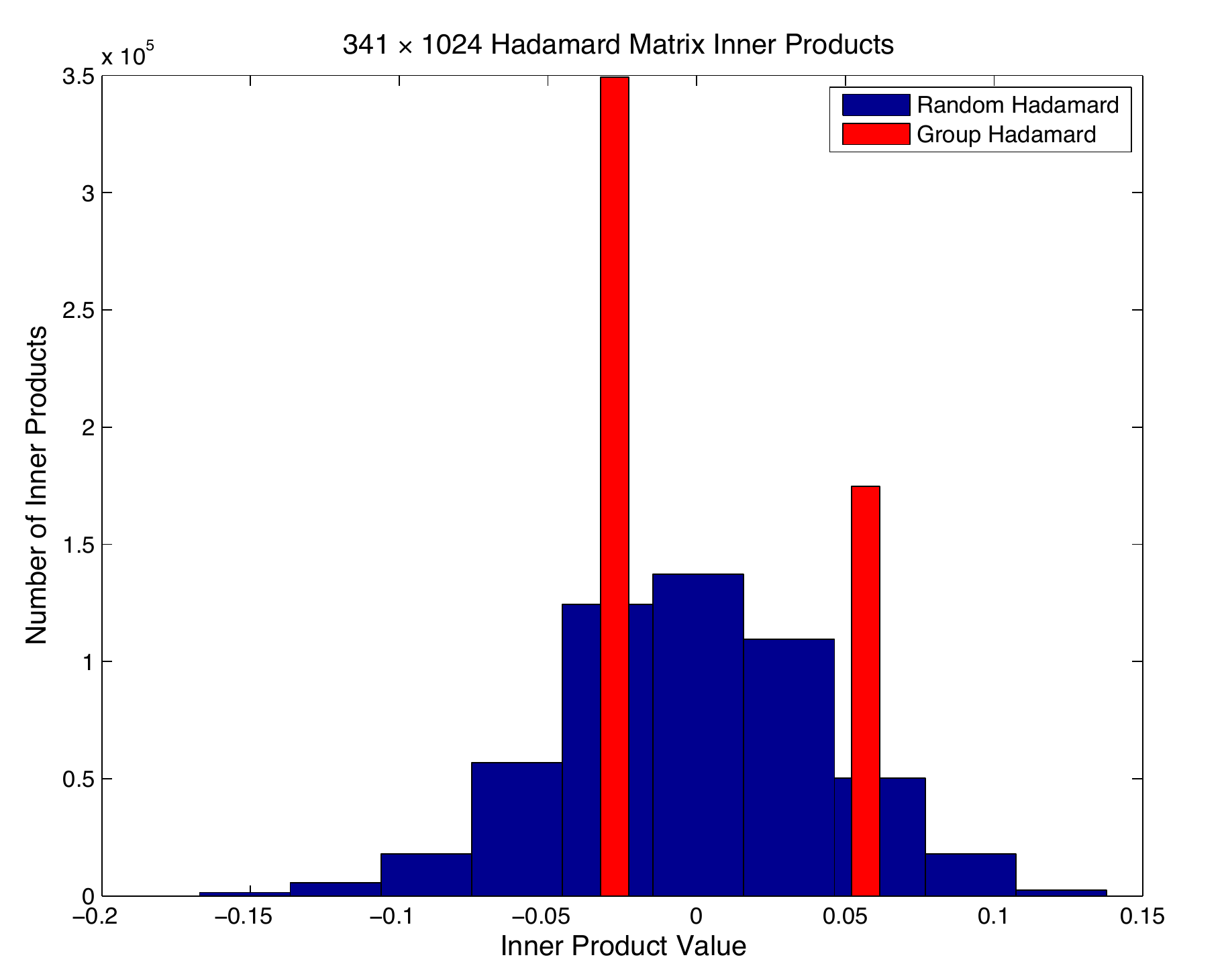}
 \caption{Superimposed histograms of the inner product values between elements of a $341 \times 1024$ frame formed from our method of selecting rows of the Hadamard matrix (red) versus selecting the rows randomly (blue). The red values are concentrated to only two points in this case, resulting in coherence closer to zero. }
 \label{fig:Hadinnprodhist_341_1024}
\end{center}
\end{figure}

\begin{table}[ht]
\caption{Coherences for Random vs. Group Hadamards} 
\centering 
\begin{tabular}{c p{1.5 cm}  p{1.8cm}   p{1.5cm} c }
\hline \hline 
$(n, m)$ &  Random Hadamard & Group Hadamard & $\sqrt{\frac{n - m}{m (n-1)}}$ \\ [0.5ex] 
\hline 
(256, 51) &  .3725 & .2549 & .1256 \\ 
(256, 85) &  .2941 & .1294 & .0888 \\ 
(512, 73) &  .3425 & .2329 & .1085 \\ 
(1024, 341) & .2023 & .0616 & .0442 \\ 
(4096, 455) & .1868 & .1253 & .0442 \\ 
\hline 
\end{tabular} 
\label{table:Hadamardcoherencelist} 
\end{table}

\subsection{Difference Sets} 
\label{sec:differencesets} 

On one final note, we point out that in certain cases the group $A$ forms a \textit{difference set} in $\F_{p^r}$, that is, each nonzero element of $\F_{p^r}$ occurs as a difference $a_i - a_j$ the same number of times. In this case, our frames yield examples of those constructed \cite{Giannakis} and \cite{Ding}: 

\begin{theorem} 
The columns of $\M$ in (\ref{eqn:Mfieldmat}) form a tight equiangular frame if and only if the elements in $A = \{a_i\}_{i = 1}^m$ form a difference set in $\F_{p^r}$. In this case, the coherence of $\M$ achieves the Welch bound. In particular, our construction yields a difference set when $\frac{p^r-1}{m} = 2$ and $m$ is odd. 
\label{thm:fieldframediffsets} 
\end{theorem} 

\begin{proof} 
Again, this follows from the arguments in \cite{Giannakis} and \cite{Ding} (see Theorem 3 of the latter). When $\frac{p^r - 1}{m} = 2$, $A$ is the group of squared elements in $\F_{p^r}^\times$, which is a well-known difference set when $p^r \equiv 3 \mod 4$ (an example of what is called a ``Paley difference set''). \cite{Wallis} This is precisely the case when $m$ is odd. 
\end{proof}

Unfortunately, the Hadamard frames we constructed in the previous section cannot satisfy the condition $\frac{p^r-1}{m} = 2$, since they require that $p = 2$. We can, however, use our construction to produce tight, equiangular frames whose entries are from an alphabet only of size three---the third roots of unity: 

\begin{corollary} 
Let $p \equiv 3 \mod 4$ be a prime, $r$ an odd integer, and set $m := \frac{p^r - 1}{2}$. Choose the set $A = \{a_i\}_{i = 1}^m$ to be the unique subgroup of $\F_{p^r}$ of size $m$. Then the columns of $\M$ in (\ref{eqn:Mfieldmat}) form a tight equiangular frame whose entries are each one of the distinct $p^{th}$ roots of unity. In particular, when $p = 3$, the entries of $\M$ come from an alphabet of size three. 
\label{cor:finitefieldkappaequalstwo} 
\end{corollary} 

\begin{proof} 
Since $r$ is odd, we have $p^r \equiv 3 \mod 4$, so the set $A$ forms a Paley difference set as mentioned in the proof of Theorem \ref{thm:fieldframediffsets}. Thus the columns of $M$ form a tight, equiangular frame whose elements are integer powers of $\omega = e^{2 \pi i / p}$, i.e., the $p^{th}$ roots of unity.  
\end{proof}

In Table \ref{table:finitefieldcoherencelist}, we list the coherences of several of the tight, equiangular frames arising from Corollary \ref{cor:finitefieldkappaequalstwo}, and compare the coherence to when the matrix $\M$ in (\ref{eqn:Mfieldmat}) is formed by randomly choosing the elements $\{a_i\}_{i = 1}^m$. As expected, our frames consistently have lower coherence, in this case meeting the Welch bound.

\begin{table}[ht]
\caption{Coherences for Random vs. Group Matrices with Small Alphabets, $m = \frac{n-1}{2}$} 
\centering 
\begin{tabular}{c p{1.5 cm}  p{1.8cm}   p{1.5cm} c }
\hline \hline 
$n$ &  Random & Group & $\sqrt{\frac{n - m}{m (n-1)}}$ \\ [0.5ex] 
\hline 
$3^3$ &  .3353 & .2035 & .2035 \\ 
$3^5$ &  .1577 & .0645 & .0645\\ 
$3^7$ &  .0509 & .0214 & .0214 \\ 
$7^3$ & .1110 & .0542 & .0542 \\ 
$11^3$ & .0674 & .0274 & .0274 \\ 
\hline 
\end{tabular} 
\label{table:finitefieldcoherencelist} 
\caption*{\small{Coherences of $m \times n$ frame matrices formed from rows of the group Fourier matrices for the finite fields $\F_q$, $q = n$. We compare choosing the rows randomly with using the group method from Section \ref{sec:vectorspacesfinitefields}, which produces tight, equiangular frames by Corollary \ref{cor:finitefieldkappaequalstwo}. When $n = p^r$, the matrix entries are $p^{th}$ roots of unity. }}
\end{table} 


\section{Frames from Special Linear Groups} 

\begin{table*}[t]
\caption{Character Table of $SL_2(\F_q)$, $q$ even} 
\centering 
\begin{tabular}{p{3.5cm} c c c c }
\hline \hline 
Class Representative: & $\begin{bmatrix} 1 & 0 \\ 0 & 1 \end{bmatrix}$ & $\begin{bmatrix} 1 & 1 \\ 0 & 1 \end{bmatrix}$  & $\begin{bmatrix} c & 0 \\ 0 & c^{-1} \end{bmatrix}$ & \hspace{-10pt}$B \begin{bmatrix} s & 0 \\ 0 & s^{-1} \end{bmatrix} B^{-1}$ \\ 
\hline 
No. of such classes: & 1 & 1 & $\frac{1}{2} (q - 2)$ & $\frac{1}{2} q$  \\ 
Size of class: & 1 & $q^2 - 1$ & $q(q+1)$ & $q(q-1)$\\ [0.5ex] 
\hline 
\hspace{30pt} $1_G$ & 1 & 1 & 1 & 1 \\ 
\hspace{30pt} $St_G$ & $q$ & 0 & 1 & $-1$ \\ 
\hspace{30pt} $\rho_\chi$ & $q+1$ & 1 & $\chi(c) + \chi(c^{-1})$ & 0 \\
\hspace{30pt} $\pi_\eta$ & $q-1$ & $-1$ & 0 & \hspace{-10pt} $-\eta(s) - \eta(s^{-1})$ \\ 

\hline 
\end{tabular} 
\caption*{\small{Here, $c \in \F_q$ and $s \in \F_{q^2}$, where $s$ is an element of norm 1. $B$ is an invertible matrix with entries in $\F_{q^2}$. }} 
\label{table:SL2qeven} 
\end{table*}  

To show how our framework can be applied to more complicated groups, we will demonstrate how to obtain frames with low coherence in the case where $G$ is the special linear group $SL_2(\F_q)$. 
Frames of this type were discussed in \cite{Thill_ICASSP2014}. 
This matrix group is easy to describe, but it is nonabelian and has irreducible representations of degree greater than 1, hence will be interesting for our purposes. 

Let $\F_q$ be the finite field containing $q$ elements, where $q$ is some integral power of a prime number. 
Then $SL_2(\F_q)$ is the set of $2 \times 2$ determinant-1 matrices with entries in $\F_q$, 
$$SL_2(\F_q) := \left\{ \begin{bmatrix} a & b \\ c & d \end{bmatrix}~|~ a, b, c, d \in \F_q, ~ ad - bc = 1 \right\}. $$ 
It is not difficult to check that the size of this group is $|SL_2(\F_q)| = q(q+1)(q-1)$.

Table \ref{table:SL2qeven} is the character table of $SL_2(\F_q)$ for when $q$ is even (a power of $2$). 
As we can see, in this case the matrices fall into four types of conjugacy classes based on how they diagonalize. 
The first is simply the identity matrix, $\begin{bmatrix} 1 & 0 \\ 0 & 1 \end{bmatrix}$. 
The second consists of the matrices that are not diagonalizable, and have the Jordan canonical form $\begin{bmatrix} 1 & 1 \\ 0 & 1 \end{bmatrix}$. 
These first two conjugacy classes contain all the matrices in $SL_2(\F_q)$ with repeated eigenvalues of 1.

Each conjugacy class of the third type has a representative which is a diagonal matrix: $\begin{bmatrix} c & 0 \\ 0 & c^{-1} \end{bmatrix}$, where $c \in \F_q \setminus \{0, 1\}$. 
Since the diagonal matrices $\text{diag}(c, c^{-1})$ and $\text{diag}(c^{-1}, c)$ are conjugate to each other, there are $\frac{1}{2}(q-2)$ such classes.  

The fourth type of conjugacy class consists of matrices whose eigenvalues do not lie in $\F_{q}$. 
These are the matrices that take the form $B \begin{bmatrix} s & 0 \\ 0 & s^{-1} \end{bmatrix} B^{-1}$, where $B \in SL_2(\F_{q^2})$ and $s \in \F_{q^2}$ is one of the \textit{norm-1} elements of $\F_{q^2} \setminus \F_q$, that is, $s^{q+1} = 1$. 
Note that here $\F_{q^2}$ is the finite field of $q^2$ elements, which contains $\F_q$ as a subfield. 
There are $q+1$ elements of $\F_{q^2}$ which satisfy the equation $s^{q+1} = 1$. 
Of these, the only element lying in $\F_q$ is 1, and the remaining $q$ lie in $\F_{q^2} \setminus \F_q$. 
As in the previous case, these $q$ elements pair up to represent a total of $q/2$ distinct conjugacy classes of the fourth type.

There are four types of characters of $SL_2(\F_q)$ for $q$ even, arising as a consequence of the four types of conjugacy classes. 
The interested reader can refer to \cite{Tanaka, Reeder, Prasad} to learn in depth how these characters come about, but for now we will give brief descriptions. 
The first two characters both correspond to degree-1 representations. 
They include the character of the identity representation $1_G$, which maps every element to 1, and that of the Steinberg representation $St_G$, which maps elements of the various conjugacy classes to the values shown in Table \ref{table:SL2qeven}. 
For our purposes, the third and fourth types of characters in the last two rows of the table are of greater interest. 
The third corresponds to what is called an \textit{induced representation}, denoted here as $\rho_\chi$. 
It is a degree-$(q+1)$ representation built from an underlying nontrivial degree-1 representation $\chi$ of the multiplicative group $\F_q^\times$, a cyclic group of size $q-1$. 
If $\tilde{c}$ is a cyclic generator for $\F_{q}^\times$ (so that every element can be written as a power of $\tilde{c}$), and we set $\omega_{-} = e^\frac{2 \pi i}{q-1}$, then $\chi$ is a function of the form $\chi(\tilde{c}^\ell) = \omega_{-}^{a \ell}$, for some fixed $a \in \{1, 2, ..., q-2\}$. 
(It is required that $a$ be nonzero modulo $q-1$ in order for $\rho_{\chi}$ to be irreducible.) 

The final type of character, denoted $\pi_\eta$, corresponds to a degree-$(q-1)$ \textit{cuspidal representation}. 
A cuspidal representation is constructed from a degree-1 representation $\eta$ of the set of norm-1 elements of $\F_{q^2}$, which is a cyclic multiplicative group of size $q+1$. 
Given a cyclic generator $\tilde{s}$ for this group, and setting $\omega_{+} = e^{\frac{2 \pi i}{q+1}}$, then $\eta$ will take the form $\eta(\tilde{s}^\ell) = \omega_{+}^{h \ell}$, where $h$ is some fixed integer in the set $\{1, 2, ..., q\}$. (Again we require $h \not \equiv 0 \mod q+1$ for irreducibility of $\pi_\eta$.

\subsection{Frames from Induced and Cuspidal Representations} 
\label{sec:inducedreps} 
We can now use our previous results to design low-coherence frames in the form of $\mathcal{F}$ in (\ref{eqn:matrixF}) using the characters of $SL_2(\F_{q})$ for $q$ even. 
We emphasize that while explicitly writing out our frame vectors can be cumbersome and requires a certain amount of work in its own right, we will find that \textit{identifying} which representations to use will be quick, as will computing the coherence of the resulting frame.  

We will first focus our attention on only the induced representations. 
For convenience, we will write $\chi_a$ and $\rho_a$ respectively for the representations $\chi$ and  $\rho_\chi$ where $\chi(\tilde{c}) = \omega_{-}^{a}$. 
It remains to identify a suitable group $A$ of automorphisms of $SL_2(\F_q)$ under which we can take the image of an induced representation to construct our frames, as prescribed by Theorem \ref{thm:generalizedfourierframebound}. 
In the last section, when our group was just the additive group of a finite field $\F_q$, our automorphisms corresponded to the nonzero field elements which formed the cyclic multiplicative group $\F_q^\times$. 
These automorphisms were well-described and easy to work with. 
It turns out that each automorphism $\varphi$ of $\F_q$ induces an automorphism of $SL_2(\F_q)$ by simply applying $\varphi$ to the entries of the $2 \times 2$ matrices in the special linear group: 
\begin{align} 
\varphi\left( \begin{bmatrix} a & b \\ c & d \end{bmatrix} \right) := \left( \begin{bmatrix} \varphi(a) & \varphi(b) \\ \varphi(c) & \varphi(d) \end{bmatrix} \right). 
\end{align} 
This observation enables us to continue working with the automorphisms of $\F_q$, so we can again choose $A$ to be a subgroup of $\F_q^\times$. 
If $a' \in A \le \F_q^\times$, then as an automorphism $a'$ acts on $\rho_a$ as 
\begin{align} 
a' \cdot \rho_a = \rho_{a' \cdot a}. 
\end{align} 
Thus, it would be natural to choose for $A$ to act on the representation $\rho_1$, so that the images under $A$ will be the representations $\{\rho_{a}~|~a \in A\}$. For the sake of simplicity, we will set $K = 1$ and $H = A$ in our Theorem \ref{thm:generalizedfourierframebound} notation. 

One caveat that we now face by choosing this set of automorphisms is the following: notice that each element of $A$ fixes the element $u = \begin{bmatrix} 1 & 1 \\ 0 & 1 \end{bmatrix} \in SL_2(\F_q)$, which means that there is a size-1 orbit $KA(u)$. This means that the bound we gave in Equation (\ref{eqn:generalizedfourierframebound}) of Theorem \ref{thm:generalizedfourierframebound} will be somewhat ineffective. We can get around this problem by noticing that from Equations (\ref{eqn:sumhgalphaismu1}) and (\ref{eqn:sumhgalphaismu2}), the magnitude of the largest inner product will still be small as long as the inner product corresponding to $u$ is small in magnitude. We quickly see this to be the case based on the equation for the inner product given in (\ref{eqn:daginnerprod}) and the fact that, from Table \ref{table:SL2qeven}, the character values $\rho_a(u)$ are all equal to 1, a relatively small constant. We will give an explicit formula for the inner product corresponding to $u$ in Equation (\ref{eqn:SL2uinnerprod}), and after normalizing our frame elements (dividing the inner product by the squared norm of a frame element) this inner product becomes very small as $q$ grows.


Since we are working with such a familiar set of automorphisms $A$, we would like to exploit some of the tools we developed for our frames constructed from finite fields. Consider choosing $q$ such that $q-1$ is some prime $p$. 
In this case, $\chi_a$ is simply a representation of the cyclic group $\Z/p\Z$, which is isomorphic to the additive group of the field $\F_p$.
From the preceding sections, we already have powerful tools at our disposal for bounding certain sums of these characters.  
Since the character $\chi_a$ appears in the main part of the character $\rho_a$ (as shown in Table \ref{table:SL2qeven}), we would like to apply these tools to bound sums of the $\rho_a$ as well. 
This will allow us to use our bounds from Theorem \ref{thm:groupframebound} to obtain even tighter bounds on coherence than those we could obtain from Theorem \ref{thm:generalizedfourierframebound}. 

Intuitively, if we take $m$ to be a divisor of $p-1$, and let $A = \{a_1, ..., a_m\}$ be the unique size-$m$ subgroup of $(\Z/p\Z)^\times$ (explicitly the set $\{1, ..., p-1\}$ with multiplication modulo $p$), then we should achieve frames with low coherence by using $A$ to choose the representations $\rho_{a_i}$ to use in our frame matrix $\mathcal{F}$ from (\ref{eqn:matrixF}). 
Note that from our previous notation, 

Now, notice that based on Table \ref{table:SL2qeven}, the characters corresponding to $\rho_a$ and $\rho_{-a}$ are the same (where $-a$ is taken modulo $p$). 
This indicates that $\rho_a$ and $\rho_{-a}$ are in fact equivalent representations. 
If $-1$ is contained in $A$ and is not equivalent to $1$ in $\Z/p\Z$ (which is always the case when $q$ is even, since $p \ne 2$), then for each $a_i \in A$ we also have $-a_i \in A$, and $-a_i \not \equiv a_i$ in $\Z/p\Z$. 
In this case, the set of chosen representations $\{\rho_a~|~a \in A\}$ has repetition, and using these representations as the rows of $\mathcal{F}$ would yield repeated rows of the Group Fourier Matrix of $SL_2(\F_{q})$, and hence would not produce a tight frame (based on Theorem \ref{thm:groupfouriertightframes}). 
More importantly for our purposes, the resulting frame would not fit our criteria from Theorem \ref{thm:generalizedfourierframebound}, which means we could not use the tools we have built to bound its coherence. 
Therefore, if $-1$ lies in the unique subgroup of $(\Z/p\Z)^\times$ of size $m$, we must choose $A$ slightly differently. 
First, let us explicitly describe how the size-$m$ subgroup decomposes into pairs $\{a, -a\}$: 

\begin{lemma} 
Let $q=2^d$ for some positive integer $d$, such that $p = q-1$ is a prime. 
Take a divisor $m$ of $p-1$, and let $A_m$ be the unique size-$m$ subgroup of $(\Z/p\Z)^\times$. 
Then $A_m$ contains $-1$ if and only if $m$ is even. 
In this case, $m/2$ is odd, and $A_m = A_{m/2} \cup -A_{m/2}$ where $A_{m/2}$ is the unique size-$\frac{m}{2}$ subgroup and $-A_{m/2} = \{-a ~|~a \in A_{m/2}\}$.  
\label{lem:meven} 
\end{lemma} 

\begin{proof} 
Since $p$ is necessarily odd, $-1$ generates the unique size-2 subgroup of $(\Z/p\Z)^\times$, and $A_m$ contains this subgroup if and only if its size $m$ is even. 

Since $p = q-1$ is prime, then writing $q$ in the form  $2^d$ for some integer $d$, we must have $d>1$.
In this case, $m$ is a divisor of $q-2 = 2(2^{d-1} - 1)$. 
In this form, it is clear that $q-2$ can never be divisible by 4 (since the factor $(2^{d-1} - 1)$ is odd), so neither can its divisor $m$. 
Thus, if $m$ is even, $m/2$ must be odd, so $-1 \notin A_{m/2}$. 
As a result, for any $a \in A_{m/2}$, we must have $-a \in A_m \setminus A_{m/2}$ (since $A_{m/2}$ is a subgroup of $A_m$). 
By comparing sizes, we see that $A_m$ must be equal to the union $A_{m/2} \cup -A_{m/2}$. 
\end{proof}

From Lemma \ref{lem:meven}, we see that when $m$ is an even divisor of $p-1$, the obvious candidate for the group $A$ is the unique size-$\frac{m}{2}$ subgroup of $(\Z/p\Z)^\times$, which will ensure that $-1$ is not in $A$ and that our resulting frame is tight. 
With this in mind, we will simply assume that we choose $m$ to be odd. 
The following theorem uses our previous results on frames constructed from finite fields to give a bound on the coherence of the frames we can construct from the induced representations of $SL_2(\F_q)$, for $q$ even.

\begin{theorem}
Take $q$ a power of 2 such that $q-1$ is a prime $p$, and let $m$ be an odd divisor of $p-1$ and $\kappa = \frac{p-1}{2m}$. Let $A = \{a_1, ..., a_m\}$ be the unique subgroup of $(\Z/p\Z)^\times$ of size $m$, and form $\mathcal{F}$ (as in (\ref{eqn:matrixF})) from the induced representations $\rho_{a_i}$. Then the coherence $\mu_{\mathcal{F}}$ of $\mathcal{F}$ is bounded by 
\begin{align} 
\mu_{\mathcal{F}} \le \frac{1}{q+1} \max \left(  1 , ~ \frac{1}{\kappa} \left( (\kappa-1) \sqrt{\frac{1}{2m} \left(\kappa + \frac{1}{2m}\right)} + \frac{1}{2m}\right) \right).  
\end{align} 
\label{thm:SL2cohbound} 
\end{theorem}



\vspace{-10pt}

\begin{proof} 
From Equation (\ref{eqn:groupFTinnerprodstep2}) and Table \ref{table:SL2qeven}, we see that the only nontrivial inner products between the columns of $\mathcal{F}$ are those corresponding to the conjugacy classes represented by $u := \begin{bmatrix} 1 & 1 \\ 0 & 1 \end{bmatrix} \in SL_2(\F_q)$ and $w_{\ell} := \begin{bmatrix} \tilde{c}^\ell & 0 \\ 0 & \tilde{c}^{-\ell} \end{bmatrix} \in SL_2(\F_2)$ for $\ell \in \{1, ..., q-2\}$. 
These inner products are: 
\begin{align} 
u: \hspace{15pt} \sum_{i = 1}^m d_i \chi_{\rho_{a_i}}(u) & =  m(q+1) \label{eqn:SL2uinnerprod} \\
& \\ 
w_\ell:  \hspace{10pt} \sum_{i = 1}^m d_i \chi_{\rho_{a_i}}(w_\ell) & = \sum_{i = 1}^m (q+1)\cdot (\chi_{a_i}(\tilde{c}^\ell) + \chi_{a_i}(\tilde{c}^{-\ell})) \\ 
 & = (q+1)   \sum_{i = 1}^m (\omega_{-}^{\ell a_i} + \omega_{-}^{-\ell a_i}). \label{eqn:asum} 
\end{align} 

From Lemma \ref{lem:meven} and the fact that $m$ is odd by assumption, we can see that the union $A \cup -A = \{\pm a_1, ..., \pm a_m\}$ is actually the unique subgroup of $(\Z/p\Z)^\times$ of size $2m$. 
If we denote this subgroup as $A_{2m}$, then we can write the sum in (\ref{eqn:asum}) in the form 
\begin{align} 
\sum_{i = 1}^m (\omega_{-}^{\ell a_i} + \omega_{-}^{-\ell a_i}) = \sum_{a \in A_{2m}} \omega_{-}^{\ell a}. 
\end{align} 
But this is just a scaled version of one of our original inner products between the elements of the harmonic frames that we constructed in Theorems \ref{thm:cosetinnerproducts} and \ref{thm:coherenceupperbound}, so we can use Theorem \ref{thm:coherenceupperbound} to bound its magnitude. 

To complete the proof, we simply need to take the maximum of the inner product magnitudes corresponding to the elements $u$ and $w_\ell$. 
This maximum becomes scaled after we normalize the columns of $\mathcal{F}$ by $\sqrt{m(q+1)^2}$, where we obtain the column norm from Equation (\ref{eqn:GFMcolnorm}) and the fact that the induced representations are $(q+1)$-dimensional. 
\end{proof}

\begin{figure}
\begin{center}
 \includegraphics[height = 230pt, width=275pt]{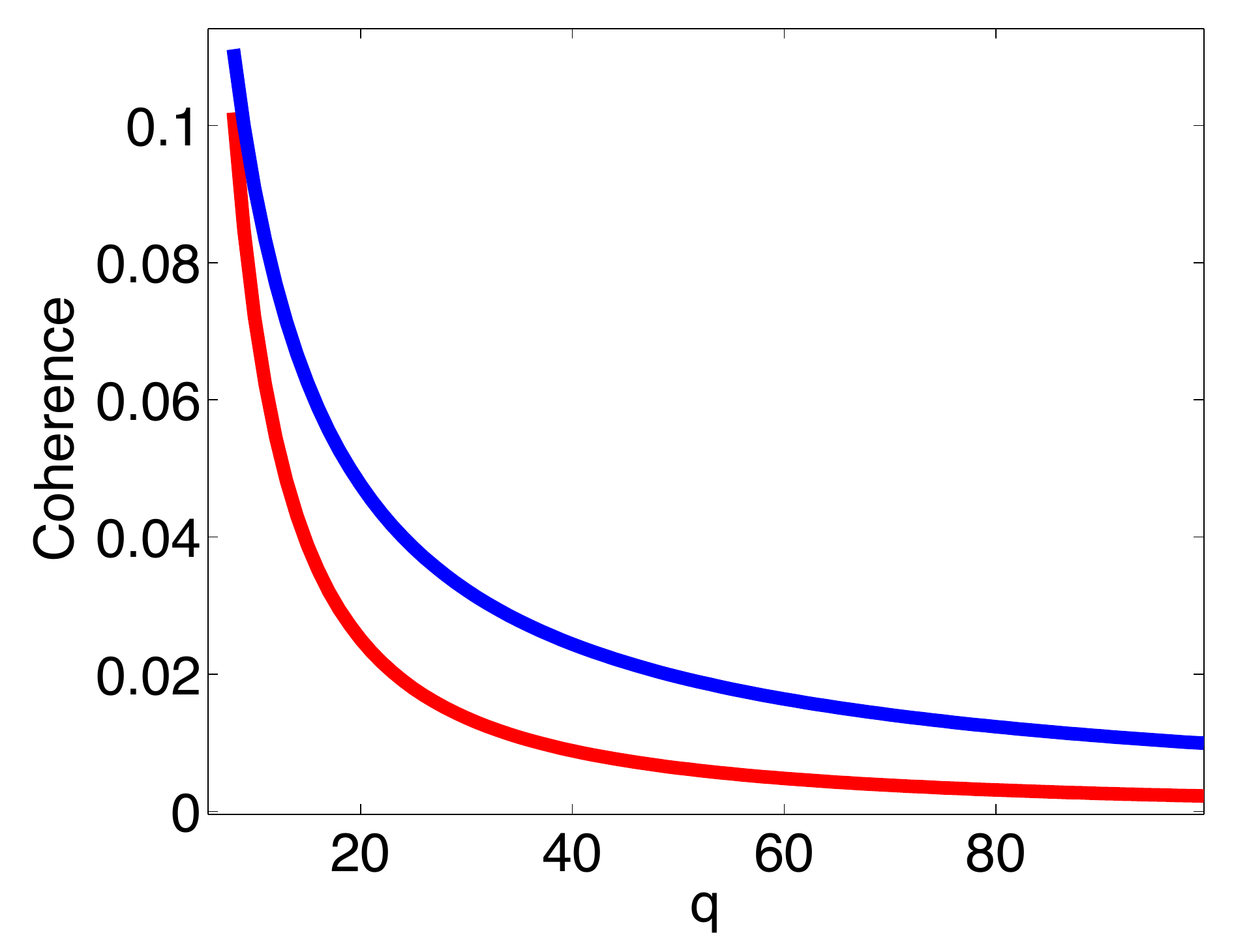}
\caption{Comparison of the Welch lower bound on coherence with the upper bound given by Theorem \ref{thm:SL2cohbound} for frames constructed from the induced representations of $SL_2(\F_q)$, for $q$ a power of 2 such that $q-1$ is prime. The number of frame vectors is $|SL_2(\F_{q})| = q(q+1)(q-1)$, which are $m(q+1)^2$-dimensional. Here, we have fixed $\kappa := \frac{q-2}{2m} = 3$. } 
 \label{fig:SL2inducedboundsr3}
\end{center}
\end{figure}

\begin{table}[ht]
\caption{$SL_2(\F_q)$ vs. Gaussian Frame Coherences}
\centering
\begin{tabular}{c c c c }
\hline \hline
Frame Dimensions & $SL_2(\F_q)$  & Random Gaussian & Welch Bound  \\
\hline
$25 \times 60$& .2000 & .5214 & .1540 \\
$81 \times 504$& .2002 & .3482 & .1019 \\
$243 \times 504$ & .1111 & .2274 & .0462 \\

\hline
\end{tabular}
\label{table:SL2inducedframecoherences}
\end{table}

Theorem \ref{thm:SL2cohbound} gives us a recipe for constructing low-coherence frames from the induced representations of $SL_2(\F_q)$ for $q$ even. 
These frames will consist of $q(q+1)(q-1)$ vectors (one for each element of $SL_2(\F_q)$) which are $m(q+1)^2$-dimensional. 
Figure \ref{fig:SL2inducedboundsr3} shows how our upper bound from the theorem comes decently close to the Welch lower bound on coherence. 
In table \ref{table:SL2inducedframecoherences}, we provide some explicit values of our frames' coherence, and for comparison we have included the coherence of frames of the same dimensions and number of elements whose coordinates are chosen independently from a Gaussian distribution. 
While the frame matrix $\mathcal{F}$ can be concretely written out using the explicit forms of the representations given in \cite{Tanaka, Reeder, Prasad}, we will omit this process since we have already described it in depth and since these particular frames tend to have rather large dimensions. 
We remark that we can obtain similar results using a parallel construction of $\mathcal{F}$ with only cuspidal representations $\pi_\eta$, which works when $q+1$ is prime.

\section{Satisfying the Coherence Property and the Strong Coherence Property} 
\label{sec:strongcoherence} 
A closely related quantity to the coherence of a frame $\{f_i\}_{i = 1}^n$ in $\C^m$ is the \textit{average coherence} $\nu$, defined as 
\begin{align} 
\nu = \frac{1}{n-1} \max_{i \in [n]} \left| \sum_{j \ne i} \langle f_i. f_j \rangle \right|. 
\end{align} 
When discussing the average coherence, the usual quantity $\mu$ is sometimes referred to as the \textit{worst-case} coherence. \cite{BajwaCalderbankJafarpour} and \cite{MixonBajwaCalderbank} use the average coherence to describe the following properties of certain frames: 

\begin{definition} 
A frame $\{f_i\}_{i = 1}^n$ in $\C^m$ with average coherence $\nu$ and worst-case coherence $\mu$ is said to satisfy the \textit{Coherence Property} if 
\begin{enumerate} 
\item $\mu \le \frac{0.1}{\sqrt{2 \log n}}$, and 
\item $\nu \le \frac{\mu}{\sqrt{m}}$. 
\end{enumerate} 
It satisfies the \textit{Strong Coherence Property} if 
\begin{enumerate} 
\item $\mu \le \frac{1}{164 \log n}$, and 
\item $\nu \le \frac{\mu}{\sqrt{m}}$. 
\end{enumerate} 
\end{definition} 

These works also give theoretical guarantees on the sparse-signal-recovery abilities of frames satisfying these properties. In particular, they discuss the One-Step Thresholding (OST) algorithm described in \cite{BajwaCalderbankJafarpour}. If $F \in \C^{m \times n}$ has columns which form a unit-norm frame, $x \in \C^{n \times 1}$ is a sparse signal, and $e \in \C^{n \times 1}$ is a noise vector, OST produces an estimate $\hat{x}$ for $x$ given $y := Fx + e$. If $F$ satisfies the coherence property, \cite{BajwaCalderbankJafarpour} finds regimes in which the support of $\hat{x}$ is equal to that of $x$ with high probability. If $F$ further satisfies the \textit{strong} coherence property, \cite{MixonBajwaCalderbank} further provides high-probability bounds on the error $||x - \hat{x}||_2$. In the absence of an error vector $e$, \cite{BajwaCalderbankJafarpour} also finds cases where $\hat{x}$ is identically equal to $x$ with high probability.

It turns out that we can explicitly compute the average coherence of our frames from Theorem \ref{thm:generalizedfourierframebound}, and indeed any group frame constructed from a set of distinct irreducible representations of the same degree:

\begin{theorem} 
Let $G$ be a finite group of size $n$ and $\rho_1, ..., \rho_m$ a set of distinct nontrivial degree-$d$ irreducible representations of $G$. Then the columns of the matrix ${\M = \sqrt{d} [{\rm vec}\rho_i(g_j)] \in \C^{md \times n}}$ from (\ref{eqn:Fmatmrows}) form a frame with average coherence $\nu = \frac{1}{n-1}$. If $\mu$ is the worst-case coherence of $\M$, then $\nu \le \frac{\mu}{\sqrt{md}}$ provided that $n \ge 2md$. 
\end{theorem} 

\begin{proof} 
From equations (\ref{eqn:groupFTinnerprodstep2}) and (\ref{eqn:GFMcolnorm}), we have that after normalizing the columns of $\M$, the inner product between the $i^{th}$ and $j^{th}$ columns is 
\begin{align}
\frac{\sum_{t = 1}^m d \cdot \text{vec}(\rho_t(g_i))^* \text{vec}(\rho_t(g_j))}{m d^2} & = \frac{1}{md} \sum_{t = 1}^m  \chi_t(g_i^{-1} g_j). 
\end{align}
Then the average coherence of $\M$ (after normalizing the columns) becomes 
\begin{align} 
\nu = \frac{1}{n-1} \max_{i \in [n]} \left| \sum_{j \ne i} \langle f_i, f_j \rangle \right| & = \frac{1}{md(n-1)} \max_{i \in [n]} \left| \sum_{j \ne i} \sum_{t = 1}^m  \chi_t(g_i^{-1} g_j) \right| \\
& = \frac{1}{md(n-1)} \left| \sum_{g \ne 1} \sum_{t = 1}^m  \chi_t(g) \right| \\ 
& = \frac{1}{md(n-1)} \left| \sum_{t = 1}^m \left(\sum_{g \ne 1}  \chi_t(g)\right) \right|. 
\end{align} 
Now from basic character theory (see for example \cite{Serre}), we know that for any character $\chi_t$ of a \textit{nontrivial} irreducible representation, we have the relation 
\begin{align} 
\frac{1}{|G|} \sum_{g \in G} \chi_t(g) = 0. 
\end{align} 
This is due to the orthogonality of irreducible characters, and the above sum is simply the inner product between $\chi_t$ and the trivial character. But this equation gives us 
\begin{align} 
\sum_{g \ne 1} \chi_t(g) = -\chi_t(1) = -d,
\end{align} 
since $\chi_t(1)$ is the degree of the representation $\rho_t$. Thus, 
\begin{align} 
\nu & = \frac{1}{md(n-1)} \left| \sum_{t = 1}^m (-d) \right| = \frac{md}{md(n-1)} = \frac{1}{n-1}. 
\end{align} 
Now from the Welch bound, $\mu \ge \sqrt{\frac{n - md}{md(n-1)}}$. Thus, to show that $\nu \le \frac{\mu}{\sqrt{md}}$ it is sufficient to show that $\frac{1}{n-1} \le \frac{1}{\sqrt{md}} \sqrt{\frac{n - md}{md(n-1)}}$, or equivalently that 
\begin{align} 
md \le \sqrt{(n - md)(n-1)}. \label{eqn:strongcohthmeq1} 
\end{align}
But since $n-1 \ge n - md$, we have $\sqrt{(n - md)(n-1)} \ge n-md$, so (\ref{eqn:strongcohthmeq1}) is satisfied provided that $2md \le n$. 
\end{proof}

\cite{MixonBajwaCalderbank} explored the geometry of several types of frames to see when they satisfied the Coherence and Strong Coherence Properties. In particular, they stated the following theorem: 

\begin{theorem}[\cite{MixonBajwaCalderbank}] 
Let $F$ be an $n \times n$ discrete Fourier matrix, $F_{k \ell} = e^{2 \pi i k \ell/n}$, $k, \ell = 0, ..., n-1$. Then let $M$ be the submatrix formed by randomly selecting a subset of rows of $F$, each row independently selected with probability $\frac{m}{n}$, and then normalizing the columns. If $16 \log n \le m \le \frac{n}{3}$, then with probability exceeding $1 - 4n^{-1} - n^{-2}$ the worst-case coherence of $M$ satisfies $\mu_M \le \sqrt{\frac{118(n-m)\log n}{mn}}$. 
\label{thm:randomharmonicframes} 
\end{theorem} 

In Figure \ref{fig:harmoniccoherencebounds}, we compare this bound with the bound on our harmonic frames from Theorem \ref{thm:coherenceupperbound} and the Welch lower bound on coherence, in the regimes where $m = \frac{n-1}{3}$ (i.e. $\kappa = 3$) and when $m = n^{4/5}$. In both cases, we can see that the frames from our group-based construction are guaranteed to satisfy the Coherence and Strong Coherence Properties for a wider range of values of $n$ than random harmonic frames, as suggested by Theorem \ref{thm:randomharmonicframes}.

\begin{figure}[!h]
\centering
\begin{tabular}{cccc}
\subfigure[]{\includegraphics[height = 3.25in, width = 4.5in]{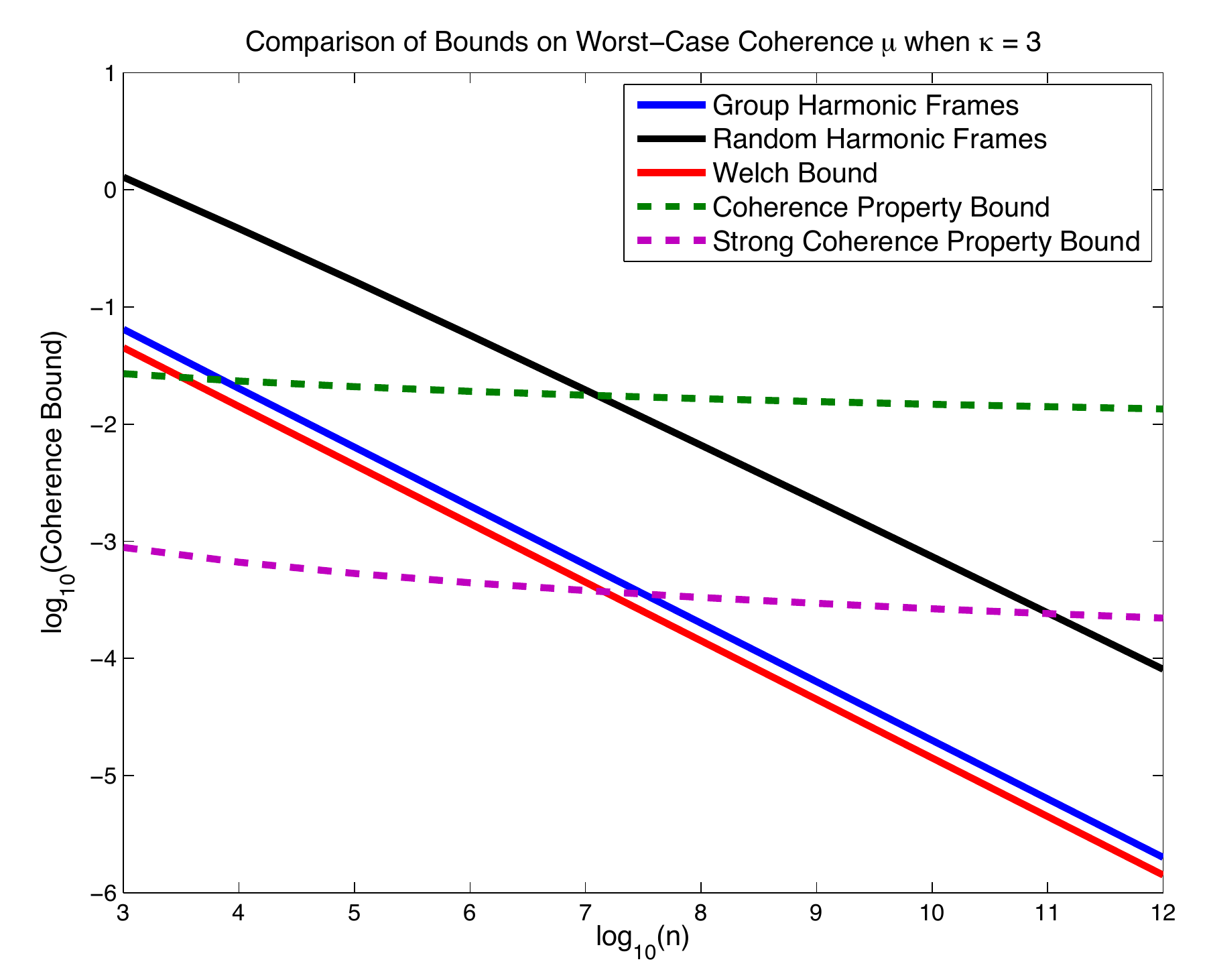} \label{subfig:kappa3}} \\ 
\subfigure[]{\includegraphics[height = 3.25in, width = 4.5in]{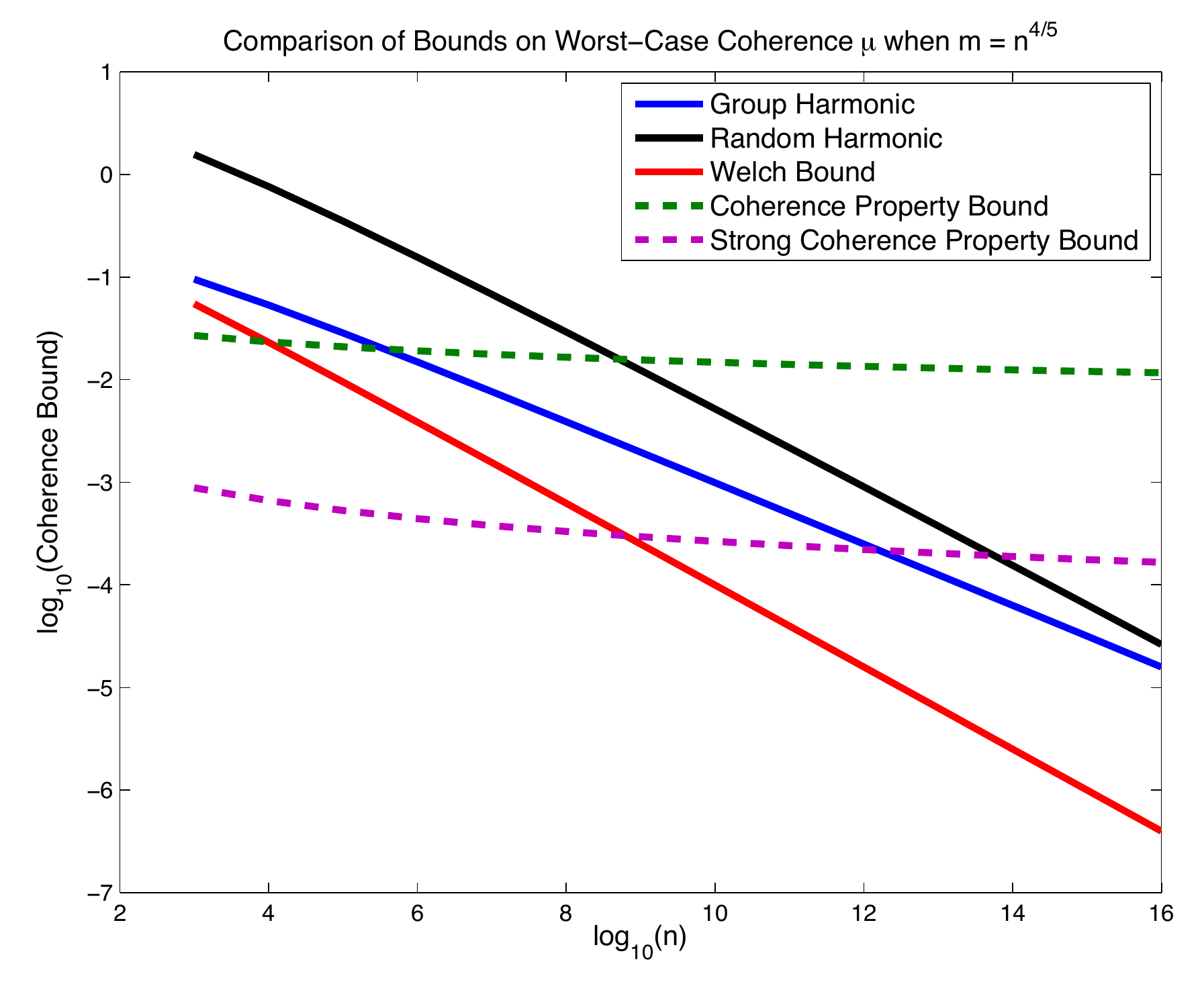} \label{subfig:meqn45}} 
\end{tabular}
\caption{Comparison of the upper bounds on the coherence of $m \times n$ harmonic frames using our group construction (from Theorem \ref{thm:coherenceupperbound}) versus choosing rows randomly from a DFT matrix (Theorem \ref{thm:randomharmonicframes}). In \ref{subfig:kappa3}, $\kappa = \frac{n-1}{m} = 3$, while in \ref{subfig:meqn45}, $m = n^{4/5}$. In both these regimes, the frames from our constructions are guaranteed to satisfy both the Coherence Property and the Strong Coherence Property for smaller dimensions than randomly chosen harmonic frames. }
\label{fig:harmoniccoherencebounds}
\end{figure}

\section{Conclusion} 
In this paper, we have generalized our methods from \cite{Thill_GroupMatrixCoherence} to yield a way to select rows of the group Fourier matrix of a finite group $G$ to produce frames with low coherence. 
By choosing the rows corresponding to the image of a representation under a subgroup of $Aut(G)$, we can reduce the number of distinct inner product values which arise between our frame elements. 
By exploiting the tightness of the resulting frames, we identified cases in which the coherence comes very close to the Welch lower bound. 

We have demonstrated that our method is particularly effective when $G$ is a subgroup or quotient of a group of matrices with entries in a finite field. 
This is a consequence of the manner in which the field automorphisms permute the elements of $G$. 
It is certainly possible that other groups of automorphisms of $G$ can lead to even better coherence when applying our method, though these remain to be explored. 

Furthermore, we emphasize that using the character table of $G$ to identify suitable representations to use in our frame allows us to avoid dealing with the explicit forms of the matrices involved in the representations. 
These matrices are often quite large in dimension and tedious to construct, particularly in the case of the special linear groups we examined in Section \ref{sec:inducedreps}. 
While exploiting the character table makes coherence calculations relatively painless, however, it is ultimately necessary to use the representation matrices to construct the actual frame vectors. 
It is desirable to find a class of groups with uncomplicated representations that allow us to build low-coherence frames in a wide variety of dimensions.

\begin{appendices} 

\section{Universal Upper Bound On Our Frame Coherence: Proof of Theorems \ref{thm:coherenceupperbound},  \ref{thm:coherenceupperboundmodd}, and \ref{thm:groupframebound} } 

\label{sec:generalkappa}

In this section, we return to the framework of Theorems \ref{thm:coherenceupperbound},  \ref{thm:coherenceupperboundmodd}, and \ref{thm:groupframebound}. Let $p$ be a prime and $r$ a a positive integer, and set our group $G$ (in Theorem \ref{thm:generalizedfourierframebound}) to be the finite field $\F_{p^r}$. 

Our frame matrix $\M$ from (\ref{eqn:Mfieldmat}) will take the form 
\begin{align} \M  = \begin{bmatrix} \omega^{\Tr{a_1x_1}} & \omega^{\Tr{a_1x_2}} & \hdots & \omega^{\Tr{a_1 x_{n}}} \\ 
						    \vdots & \vdots & \vdots & \ddots & \vdots \\ 
						    \omega^{\Tr{a_m x_1}} & \omega^{\Tr{a_m x_2}} & \hdots & \omega^{\Tr{a_m x_{n}}}\end{bmatrix}, 
\label{eqn:Mfieldmat2} 
\end{align} 
where $\omega = e^\frac{2 \pi i}{p}$ and we have expressed the elements of $\F_{p^r}$ as $\{x_i\}_{i = 1}^{n}$. In terms of powers of $x$, we may relabel these elements as $x_1 = 0$, and $x_i = x^{i-1}$, $i = 2, ..., n = p^r$. Note that with this relabeling, the first column of $\M$ is all 1's. 

As we commented before Equation (\ref{eqn:fieldinnerprod}), if $x_j - x_i = x_\ell$, the inner product between the $i^{th}$ and $j^{th}$ columns is 
\begin{align} 
\sum_{a_t} \left(\omega^{\Tr{a_t x_i}} \right)^* \left(\omega^{\Tr{a_t x_j} } \right) & = \sum_{a_t} \omega^{\Tr{a_t (x_j - x_i)}} \\
& = \sum_{a_t} \omega^{\Tr{a_t x_\ell} }. 
\label{eqn:fieldinnerprod2} 
\end{align} 

We will be making extensive use of the sums in (\ref{eqn:fieldinnerprod2}) in this section, so we will make the following definition: 

\begin{definition} 
For any $z \in \F_{p^r}$ and $A$ a subgroup of $\F_{p^r}^\times$, we will define $c_z$ to be the normalized inner product sum corresponding to $z$, that is, 
\begin{align} 
c_z = \frac{1}{m} \sum_{a \in A} \omega^{\Tr{a z}}, 
\end{align} 
where $\omega = e^{2 \pi i / p}$. 
\end{definition}

The following property of the values $c_z$ is simple, but worth establishing: 
\begin{lemma} 
For any $z \in \F_{p^r}$, we have $c_z^* = c_{-z}$. 
\end{lemma} 

\begin{proof} 
Expanding $c_z$ as a sum, we have 
\begin{align} 
(c_z)^* & = \left( \frac{1}{m} \sum_{a \in A} \omega^{\Tr{z a}} \right)^* \\ 
& = \frac{1}{m} \sum_{a \in A} \omega^{- \Tr{z a}} \\ 
& = c_{-z}.
\end{align}  
\end{proof}

Recall that the set of nonzero field elements $\F_{p^r}^\times$ is a cyclic group under multiplication, so let $x$ be a multiplicative generator. The elements of $\F_{p^r}$ can now be expressed as $\{0, 1, x, ..., x^{p^r - 1}\}$. The inner product corresponding to $0$ is simply $c_0 = 1$, which arises only when taking the inner product of a frame element with itself. The nontrivial inner products are thus $c_{x^i}$, for $i = 0, ..., p^r-1$. 

We point out that if $A$ is a size-$m$ multiplicative subgroup, it is unique (since $\F_{p^r}$ is cyclic) and is a cyclic group generated by $x^{\kappa}$, where $\kappa = \frac{p^r-1}{m}$.  
The cosets of $A$ are $A, xA, ..., x^{\kappa - 1} A$. 
One interesting observation is that elements in the same coset of $A$ give rise to the same inner product value: 

\begin{lemma} 
If $z$ is in the coset $x^i A$, then $c_{z} = c_{x^i}$. 
\label{lem:xtotheicosets} 
\end{lemma} 

\begin{proof} 
Write $z = x^i a_z$, for some $a_z \in A$. Then, 
\begin{align} 
c_z & = \sum_{a \in A} \omega^{\Tr{x^i a_z \cdot a}} \\ 
& = \sum_{a \in A} \omega^{\Tr{x^i a}}, 
\end{align} 
where $\omega = e^{2 \pi i / p}$ and the last equality follows from the fact that since $A$ is a group, multiplication by $a_z$ simply permutes its elements. 
\end{proof} 

In light of Lemma \ref{lem:xtotheicosets}, we see concretely that there is indeed only a single nontrivial inner product value for each coset of $A$, and each arises with the same multiplicity (because each coset has the same number of elements). 
Furthermore, since $\{1, x, ..., x^{\kappa-1}\}$ is a set of representatives for each of the cosets of $A$, we only need to be concerned with the values $c_{x^i}$, $i = 0, 1, ..., \kappa-1$. 
The largest absolute value of these will be the coherence.

\begin{lemma} 
The values $c_1, c_x, ..., c_{x^{\kappa}-1}$ satisfy the equation
\begin{align} 
1 + m c_1 + m c_x + ... + m c_{x^{\kappa}-1} = 0, \label{eqn:ceqn} 
\end{align}
where $m$ is the size of $A$. 
\label{lem:csumto0}
\end{lemma} 

\begin{proof} 
If we expand the sum in (\ref{eqn:ceqn}) using the fact that each of the $m$ elements $z \in x^d A$ satisfies $c_z = c_{x^d}$, and that $c_0 = 1$, we get 
\begin{align} 
1 + \sum_{d = 1}^{\kappa - 1} m c_{x^d} & = \sum_{z \in \F_{p^r}} c_z \\ 
& = \sum_{a \in A} \sum_{z \in \F_{p^r}} \omega^{\Tr{z a}}, \label{eqn:internalsumisregularrep} 
\end{align}  
where $\omega = e^\frac{2 \pi i}{p}$. 
But the function $\chi_{reg}(y) := \sum_{z \in \F_{p^r}} \omega^{\Tr{z y}}$ which arises as the internal sum in (\ref{eqn:internalsumisregularrep}) is the well-known character of the ``regular representation'' of $\F_{p^r}$, which is equal to $p^r$ if $y=0$ and $0$ otherwise \cite{Serre}. Since no elements of $A$ are 0, we see that (\ref{eqn:internalsumisregularrep}) sums to zero. 
\end{proof}

Our following work will involve taking many sums and products of field elements, and determining in which coset of $A$ they lie. While it is in general easy to determine in which coset a product lies (for example, if $z_1 \in x^{i_1} A$ and $z_2 \in x^{i_2} A$, then $z_1 z_2 \in x^{i_1 + i_2} A$), it is often not obvious in which coset a sum lies. To get around this problem, we will make use of the following quantities: 

\begin{definition} 
Given two cosets $x_1 A$ and $x_2 A$, we define the \textit{translation degree} from $x_1 A$ to $x_2 A$ to be the quantity 
\begin{align} 
N_{x_1 A, x_2 A} = \# \{z \in x_1 A~|~ 1 + z \in x_2 A\} = |1+x_1 A \cap x_2 A|. 
\end{align} 
Likewise, we define $N_{x_1 A, 0}$ and $N_{0, x_2 A}$ (the translation degrees from $x_1 A$ to $0$ and from $0$ to $x_2 A$, respectively) to be 
\begin{align} 
N_{x_1 A, 0} = |\{-1\} \cap x_1 A|, \\ 
N_{0, x_2 A} = |\{1\} \cap x_2 A|.  
\end{align} 
\end{definition}

We will quickly point out a simple property of the translation degrees: 

\begin{lemma} 
Set $H = \F_{p^r}^\times$. 
For any coset $x_0 A$, we have $$N_{x_0K, 0} + \sum_{x_i K \in H/A} N_{x_0 A, x_i A} = |x_0 A|. $$ 
In particular, if $-1 \in x_0 A$, this equation reduces to $$1 + N_{x_0 A, A} + N_{x_0 A, xA} + N_{x_0 A, x^2 A} + ... + N_{x_0 A, x^{\kappa-1} A} = m, $$ and if $-1 \notin x_0 A$, this equation becomes $$N_{x_0 A, A} + N_{x_0 A, xA} + N_{x_0 A, x^2 A} + ... + N_{x_0 A, x^{\kappa-1} A} = m. $$ 
\label{lem:sumoftransdegs} 
\end{lemma} 

\begin{proof} 
This simply follows from the observation that any of the $m$ elements of $x_0 A$, when added to 1, must either be equal to 0 or lie in exactly one of the cosets $x_i A \in H/A$. 
\end{proof}

The following lemma will be instrumental in bounding these inner product values. 

\begin{lemma} Let $\cc = [c_1, c_x, c_{x^2}, ..., c_{x^{\kappa-1}}]^T$, and let $F$ be the scaled $\kappa \times \kappa$ Fourier matrix with entries defined by $F_{ij} = \gamma^{(i - 1) (j- 1)}$, where $\gamma = e^{2 \pi i/\kappa}$. 
Then, if we let $\w := [w_1, ..., w_\kappa]^T = F \cc$ so that $w_{d + 1} = \sum_{t = 0}^{\kappa-1} \gamma^{t d} c_{x^{t}}$ for $d = 0, 1, ..., \kappa-1$, we have 

\begin{align} 
w_1 &= - \frac{1}{m}, &  \\ 
|w_i| &= \sqrt{\frac{1}{m} \left(\kappa + \frac{1}{m} \right)}, & i \ne 1. 
\end{align} 

\label{lem:wvec} 
\end{lemma}

\begin{proof} 
For any $d = 0, 1, ..., \kappa-1$, we have 

\begin{align} 
|w_{d+1}|^2 & = \left(\sum_{t = 0}^{\kappa-1} \gamma^{t d} c_{x^{t}} \right) \left( \sum_{\ell = 0}^{\kappa-1} \gamma^{\ell d} c_{x^{\ell}} \right)^* \\ 
& = \left(\sum_{t = 0}^{\kappa-1} \gamma^{t d} c_{x^{t}} \right) \left( \sum_{\ell = 0}^{\kappa-1} \gamma^{-\ell d} c_{-x^{\ell}} \right) \\ 
& = \sum_{t = 0}^{\kappa-1} \sum_{\ell = 0}^{\kappa-1} \gamma^{(t - \ell) d} c_{x^t} c_{-x^\ell}\\
& = \sum_{s = 0}^{\kappa-1} \sum_{\ell = 0}^{\kappa-1} \gamma^{s d} c_{x^{s+\ell}} c_{-x^\ell} \\ 
& = \sum_{s = 0}^{\kappa-1}  \gamma^{s d} \sum_{\ell = 0}^{\kappa-1} c_{x^{s+\ell}} c_{- x^\ell}  \label{eqn:wnormsq} 
\end{align} 

Now, we note that 

\begin{align} 
m^2 c_{x^{s + \ell}} c_{- x^{\ell}} & = \left(\sum_{a \in A} \omega^{\Tr{x^{s+ \ell} a}}\right) \left(\sum_{a' \in A} \omega^{\Tr{-x^\ell a'}} \right) \\ 
& = \sum_{a, a' \in A} \omega^{\Tr{x^{s+ \ell} a - x^{\ell} a'}} \\ 
& = \sum_{a, a' \in A} \omega^{\Tr{-x^\ell a' (1 - x^{s} a a'^{-1})}} \\ 
& = \sum_{a', a'' \in A} \omega^{\Tr{-x^\ell a' (1 - x^{s} a'')}} \\ 
& = \sum_{t = 0}^{\kappa - 1} \sum_{\substack{\{a', a'' \in A ~: \\ 1 - x^s a'' \in x^t A\}}} \omega^{\Tr{-x^\ell a' (1 - x^{s} a'')}} \\ 
\nonumber & \hspace{40pt}+ \sum_{\substack{\{a', a'' \in A ~: \\ 1 - x^s a'' = 0\}}} 1 \\ 
& = \sum_{t = 0}^{\kappa-1} N_{-x^s A, x^t A} \left( \sum_{a''' \in A} \omega^{\Tr{-x^t x^\ell a'''}} \right) \\ 
& \hspace{40pt} + \sum_{a' \in A}  N_{-x^s A, 0} \\ 
& = m \sum_{t = 0}^{\kappa-1} N_{-x^s A, x^t A} \cdot c_{-x^{t +\ell}} + m N_{-x^s A, 0} 
\end{align} 

Now we can substitute this into (\ref{eqn:wnormsq}), and we obtain: 

\begin{align} 
|w_{d+1}|^2 & = \sum_{s = 0}^{\kappa-1}  \gamma^{s d} \sum_{\ell = 0}^{\kappa-1} \left( \frac{1}{m} \left(\sum_{t = 0}^{\kappa-1} N_{-x^s A, x^t A} \cdot c_{-x^{t +\ell}} + N_{-x^s A, 0} \right) \right) \\ 
& = \sum_{s = 0}^{\kappa-1}  \gamma^{s d}   \frac{1}{m} \left(\sum_{t = 0}^{\kappa-1} N_{-x^s A, x^t A} \sum_{\ell = 0}^{\kappa-1} c_{-x^{t +\ell}} + \sum_{\ell = 0}^{\kappa-1} N_{-x^s A, 0} \right) \\ 
& = \sum_{s = 0}^{\kappa-1}  \gamma^{s d}   \frac{1}{m} \left(\sum_{t = 0}^{\kappa-1} N_{-x^s A, x^t A} \left( -\frac{1}{m} \right) + \kappa N_{-x^s A, 0} \right) \label{eqn:wnormsqeq2} \\ 
& = - \frac{1}{m^2} \sum_{s = 0}^{\kappa-1}  \gamma^{s d} \sum_{t = 0}^{\kappa-1} N_{-x^s A, x^t A}  + \frac{1}{m} \sum_{s = 0}^{\kappa-1}  \gamma^{s d} \kappa N_{-x^s A, 0} \label{eqn:wnormsqeq3} \\ 
& = - \frac{1}{m^2} \sum_{s = 0}^{\kappa-1}  \gamma^{s d} \left(m - N_{-x^s A, 0} \right) + \frac{\kappa}{m} \sum_{s = 0}^{\kappa-1}  \gamma^{s d} N_{-x^s A, 0} \label{eqn:wnormsqeq4}
\end{align} 

\noindent
where (\ref{eqn:wnormsqeq2}) follows from Equation (\ref{eqn:ceqn}), and (\ref{eqn:wnormsqeq4}) follows from Lemma \ref{lem:sumoftransdegs}. Note that $N_{-x^s A, 0}$ is equal to $1$ if $s = 0$ and equal to 0 otherwise. Thus, (\ref{eqn:wnormsqeq4}) becomes 
\begin{align} 
|w_{d+1}|^2 & = -\frac{1}{m^2} \left( (m-1) + m \sum_{s = 1}^{\kappa-1}  \gamma^{s d} \right) + \frac{\kappa}{m}.  \label{eqn:wnormsqeq5} 
\end{align} 

Now, if $d \ne 0$, then $\sum_{s = 1}^{\kappa-1}  \gamma^{s d} = -1$, and after rearranging terms we obtain 
\begin{align} 
|w_{d+1}|^2 & = \frac{1}{m} \left( \kappa + \frac{1}{m} \right).  
\end{align} 

If $d = 0$, then $\sum_{s = 1}^{\kappa-1}  \gamma^{s d} = m$, and (\ref{eqn:wnormsqeq5}) gives us $|w_{1}|^2 = \frac{1}{m^2}$. In fact, in this case, we can compute $w_1$ explicitly, since 
\begin{align} 
w_1 = \sum_{t = 0}^{\kappa-1} c_{x^{t}} = -\frac{1}{m}. 
\end{align} 

\end{proof}

We can now use Lemma \ref{lem:wvec} to bound the coherence of our frames constructed from finite fields. 

\begin{theorem} 
Let $G = \F_{p^r}$ be the finite field with elements $\{x_1, ..., x_{p^r}\}$, and $H = \F_{p^r}^\times$ the (cyclic) multiplicative group of the nonzero field elements. 
If $A$ is the unique subgroup of $H$ of size $m$, with elements $\{a_1, ..., a_m\}$, and $\M$ is the frame with columns defined in (\ref{eqn:Mfieldmat2}), then the coherence $\mu$ of $\M$ is upper-bounded by 
\begin{align} 
\mu & \le \frac{1}{\kappa} \left( (\kappa-1) \sqrt{\frac{1}{m} \left(\kappa + \frac{1}{m}\right)} + \frac{1}{m} \right). 
\end{align} 
\label{thm:coherenceupperbound2}     
\end{theorem}

\begin{proof} 
The proof follows from Lemma \ref{lem:wvec}. 
Using the notation of this lemma, we may write $\cc = \frac{1}{\kappa} F^* \w$, so that 
\begin{align}
|c_{x^d}| & = \frac{1}{\kappa} \left| \sum_{j = 1}^\kappa \gamma^{d (j-1)} w_j \right| \\ 
& \le \frac{1}{\kappa}  \sum_{j = 1}^\kappa |w_j| \label{eqn:cohprooftriineq} \\  
& = \frac{1}{\kappa} \left( (\kappa-1) \sqrt{\frac{1}{m} \left(\kappa + \frac{1}{m}\right)} + \frac{1}{m} \right), \label{eqn:cohproofwveclemma} 
\end{align} 
where (\ref{eqn:cohprooftriineq}) follows from the triangle inequality and (\ref{eqn:cohproofwveclemma}) follows from Lemma \ref{lem:wvec}. Since the coherence is equal to the largest value among the $|c_{x^d}|$, $d = 0, ..., p^r - 1$, the result now follows immediately. 
\end{proof}

Recall from Theorems \ref{thm:coherenceupperboundmodd} and \ref{thm:groupframebound} that when the size $m$ of $A$ happens to be an odd integer, we can derive even tighter bounds on coherence, provided that $p$ is an odd prime. 
(Note that in our original framework of Theorem \ref{thm:cosetinnerproducts}, when $m$ was taken to be a divisor of $p-1$, the only case where $p$ could be even was when $p = 2$ and $m = 1$ in which case our frames would be 1-dimensional and trivially have coherence equal to 1.)  
We will prove this result shortly, but we first present the following equivalent condition on $A$ for when its size is even or odd. 

\begin{lemma}
Let $p$ be a prime, $r$ an integer, $m$ a divisor of $p^r - 1$, and $\kappa := \frac{p^r - 1}{m}$. 
Let $\F_{p^r}$ be the finite field with $p^r$ elements, whose multiplicative group $\F_{p^r}^\times$ has cyclic generator $x$, and let $A$ be the unique subgroup of $\F_{p^r}^\times$ of size $m$. 
Then $-1 \in A$ if and only if either $p$ or $m$ is even. If $p$ and $m$ are both odd, then $\kappa$ is even and $-1 \in x^\frac{\kappa}{2} A$. 
\label{lem:mevennegoneinA}
\end{lemma}

\begin{proof} 
If $p$ is even, that is $p = 2$, then $-1 \equiv 1$ in $\F_{p^r}$, so trivially $-1 \in A$. If $p$ is odd, then the order $m$ of $A$ is even if and only if $A$ has a subgroup of size $2$, which means there is a nontrivial element in $A$ which is a root of the polynomial $X^2 - 1$. 
The element $-1$ is the only such root. 

If both $m$ and $p$ are odd, then $p^r - 1$ must be even, hence so is $\kappa = \frac{p^r - 1}{m}$. In this case, since the square of $-1$ obviously lies in $A$ (which is equal to $x^\kappa A$), we must have $-1 \in x^\frac{\kappa}{2} A$. 
\end{proof} 


We need one more tool before we can prove our tighter bound: 

\begin{lemma} 
Let $p$, $r$, $m$, $\kappa$, $x$ and $A$ be defined as in Lemma \ref{lem:mevennegoneinA} and $\w = [w_1, ..., w_\kappa]^T$ be defined as in Lemma \ref{lem:wvec}. 
If either $p$ or $m$ is even ($-1 \in A$) then for any $d = 0, 1, ..., \kappa-1$, we have $c_{x^d} = c^*_{x^d}$, and for any $i = 2, 3, ..., \kappa$ we have $w^*_i = w_{\kappa-i + 2}$. If $p$ and $m$ are both odd ($-1 \in x^\frac{\kappa}{2} A$), then  $c_{x^d} = c^*_{x^{d + \kappa/2}}$ and $w^*_i = (-1)^{i - 1} w_{\kappa-i + 2}$. 
\label{lem:cstarwstar} 
\end{lemma}

\begin{proof} 
As usual, set $\omega = e^{2 \pi i / p}$ and $\gamma := e^{2 \pi i / \kappa}$. 
If $-1 \in A$, then multiplication by $-1$ permutes the elements of $A$, so we have 
\begin{align} 
c_{x^d}^* & = \left(\frac{1}{m} \sum_{a \in A} \omega^{x^d a}\right)^* \\ 
&= \frac{1}{m} \sum_{a \in A} \omega^{-x^d a} \\ 
&= \frac{1}{m} \sum_{a \in A} \omega^{x^d a} \\ 
& = c_{x^d}. 
\end{align} 
It follows that $c_{x^d}$ is real. Furthermore, in this case we have 

\begin{align} 
w_i^* & = \left( \sum_{j = 1}^\kappa \gamma^{(i-1)(j-1)} c_{x^{j-1}} \right)^* \\ 
& = \sum_{j = 1}^\kappa \gamma^{-(i-1)(j-1)} c_{x^{j-1}}^* \\ 
& = \sum_{j = 1}^\kappa \gamma^{(-i+1)(j-1)} c_{x^{j-1}} \\ 
& = \sum_{j = 1}^\kappa \gamma^{(\kappa-i+1)(j-1)} c_{x^{j-1}} \\ 
& = \sum_{j = 1}^\kappa \gamma^{((\kappa-i+2) - 1)(j-1)} c_{x^{j-1}} \\ 
& = w_{\kappa - i + 2}.  
\end{align} 

Now, if instead $-1 \in x^{\frac{\kappa}{2}} A$, then multiplication by $-x^\frac{\kappa}{2}$ permutes the elements of $A$, and we have 
\begin{align} 
c_{x^d} & = \frac{1}{m} \sum_{a \in A} \omega^{x^d a} \\ 
& = \frac{1}{m} \sum_{a \in A} \omega^{- x^d x^{\frac{\kappa}{2}} a} \\ 
& = \left( \frac{1}{m} \sum_{a \in A} \omega^{x^{d + \frac{\kappa}{2}} a} \right)^* \\ 
& = c_{x^{d + \kappa/2}}^*. 
\end{align} 

Also in this case, we may write 
\begin{align} 
w_i^* & = \left( \sum_{j = 1}^\kappa \gamma^{(i-1)(j-1)} c_{x^{j-1}} \right)^* \\ 
& = \sum_{j = 1}^\kappa \gamma^{-(i-1)(j-1)} c_{x^{j-1}}^* \\ 
& = \sum_{j = 1}^\kappa \gamma^{-(i-1)(j-1)} c_{x^{j-1 + \frac{\kappa}{2}}} \\ 
& = \sum_{j = 1}^\kappa \gamma^{-(i-1)(j-1 + \frac{\kappa}{2})}  \gamma^{(i-1) \frac{\kappa}{2}} c_{x^{j-1 + \frac{\kappa}{2}}} \\ 
& = \gamma^{(i-1) \frac{\kappa}{2}} \sum_{j = 1}^\kappa \gamma^{(\kappa - i + 1)(j-1 + \frac{\kappa}{2})} c_{x^{j-1 + \frac{\kappa}{2}}} \\ 
& = \gamma^{(i-1) \frac{\kappa}{2}} w_{\kappa - i + 2} \\ 
& = (-1)^{i-1} w_{\kappa - i + 2}, 
\end{align} 
where the last line follows from the fact that $\gamma^\frac{\kappa}{2} = -1$. This completes the proof of the lemma. 
\end{proof}

We are now equipped to prove the second part of Theorem \ref{thm:groupframebound}, which we restate here for convenience:


\begin{theorem} 
\label{thm:groupframebound2ndhalf} 
Let $p$ be a prime, $r$ a positive integer, $m$ a divisor of $p^r-1$, and $A = \{a_i\}_{i = 1}^m$ the unique subgroup of $\F^\times_{p^r}$ of size $m$. Then setting $\omega = e^{\frac{2 \pi i}{p}}$ and $\kappa := \frac{p^r-1}{m}$, if both $p$ and $m$ are odd, the coherence $\mu$ of our frame $\M$ in (\ref{eqn:Mfieldmat2}) satisfies 
\begin{align} 
 \mu & \le \frac{1}{\kappa} \sqrt{\left(\frac{1}{m} + \left(\frac{\kappa}{2} - 1 \right) \beta \right)^2 + \left(\frac{\kappa}{2}\right)^2 \beta^2},  
\end{align} 
where $\beta = \sqrt{\frac{1}{m} \left( \kappa + \frac{1}{m} \right) }$. 
\end{theorem}


\begin{proof} 
Since both $p$ and $m$ are odd, then from Lemma \ref{lem:mevennegoneinA} we know that $\kappa$ is even and $-1$ lies in the coset $x^{\frac{\kappa}{2}} A$. 
There is a 1-1 correspondence between the set of integers $\{1, ..., \kappa\}$ and itself which sends $j \mapsto \kappa - j + 2 \mod \kappa$, for $j = 1, ..., \kappa$. 
This mapping fixes the singletons $\{ 1 \}$ and $\{\frac{\kappa}{2} + 1\}$ and interchanges the elements in the pairs $\{j, \kappa - j + 2\}$ for $j = 2, ..., \frac{\kappa}{2}$. 

As in Lemma \ref{lem:wvec}, set $\cc = [c_1, c_x, c_{x^2}, ..., c_{x^{\kappa-1}}]^T$ and $\w := [w_1, ..., w_\kappa]^T = F \cc$, where $F$ is the scaled $\kappa \times \kappa$ Fourier matrix with entries $F_{ij} = \gamma^{(i - 1) (j- 1)}$, where $\gamma = e^{2 \pi i/\kappa}$. 
Since $-1 \in x^{\frac{\kappa}{2} A}$, then by Lemmas \ref{lem:wvec} and \ref{lem:cstarwstar} we have 
\begin{align} 
w_j \cdot w_{\kappa - j + 2} & = w_j \left(\left((-1)^{j-1}\right)^{-1} w_j^* \right) \\ 
& = (-1)^{j-1} |w_j|^2 \\ 
& = (-1)^{j-1} \beta^2 
\end{align} 
where $\beta = \sqrt{\frac{1}{m} \left(\kappa + \frac{1}{m}\right)}$.

We quickly note that given integers $i$ and $j$, the conjugate of $\gamma^{-(i-1)(j-1)}$ can be expressed as 
\begin{align} 
\left(\gamma^{-(i-1)(j-1)}\right)^* & = \gamma^{-(i-1)(-j+1)}\\ 
& = \gamma^{-(i-1)(\kappa - j + 1)} \\ 
& = \gamma^{-(i-1)( (\kappa-j + 2) - 1)}. 
\end{align}  

Note that the inverse of $F$ is $\frac{1}{\kappa} F^*$. 
From the equation $\cc = \frac{1}{\kappa} F^* \w$, we may write 
\begin{align} 
c_{x^{i-1}} & = \frac{1}{\kappa} \sum_{j = 1}^{\kappa} \gamma^{-(i-1)(j-1)} w_j. 
\end{align} 

Now we can group the terms of the summation of $c_{x^{i-1}}$ by our above subsets of indices ($\{j, \kappa - j + 2\}$ for $j = 2, .\
.., \frac{\kappa}{2}$) as follows: 
\begin{align} 
c_{x^{i-1}} & = \frac{1}{\kappa} \left[ w_1 + \gamma^{-(i-1)\frac{\kappa}{2}} w_{\frac{\kappa}{2} + 1}  + \sum_{j = 2}^{\frac{\kappa}{2}} \left( \gamma^{-(i-1)(j-1)} w_j + \gamma^{-(i-1)((\kappa - j + 2) -1)} w_{\kappa - j + 2} \right) \right]. \label{eqn:cwsummation0} 
\end{align} 

We know from Lemma \ref{lem:wvec} that $w_1 = - \frac{1}{m}$ and $|w_{\frac{\kappa}{2} + 1}| = \beta$. 
Also, $$\gamma^{-(i-1)\frac{\kappa}{2}} = (\gamma^\frac{\kappa}{2})^{-(i-1)} = (-1)^{-(i-1)} = (-1)^{i-1}, $$ and from Lemma \ref{lem:cstarwstar} we know that $w_{\frac{\kappa}{2} + 1} = (-1)^{\frac{\kappa}{2}} w_{\frac{\kappa}{2} + 1}^*$. 
Thus, if $\frac{\kappa}{2}$ is even we have that $w_{\frac{\kappa}{2} + 1}$ is purely real, so $ \gamma^{-(i-1)\frac{\kappa}{2}} w_{\frac{\kappa}{2} + 1} = \pm \beta$. 
And if $\frac{\kappa}{2}$ is odd, we have that $w_{\frac{\kappa}{2} + 1}$ is purely imaginary, in which case $ \gamma^{-(i-1)\frac{\kappa}{2}} w_{\frac{\kappa}{2} + 1} = \pm i \beta$.

From these observations and Lemma \ref{lem:cstarwstar}, we have 
\begin{align} 
\nonumber \gamma^{-(i-1)(j-1)} w_j + \gamma^{-(i-1)((\kappa - j + 2) -1)} w_{\kappa - j + 2} \\ 
= \gamma^{-(i-1)(j-1)} w_j + (-1)^{j-1} \gamma^{-(i-1)((\kappa - j + 2) -1)} w_j^* \\ 
= \gamma^{-(i-1)(j-1)} w_j + (-1)^{j-1} \left( \gamma^{-(i-1)(j-1)} w_j\right)^*. \label{eqn:gammadubeq1} 
\end{align} 
If $j$ is even (\ref{eqn:gammadubeq1}) becomes $2i \Im(\gamma^{-(i-1)(j-1)} w_j)$ and if $j$ is odd it becomes $2 \Re(\gamma^{-(i-1)(j-1)} w_j)$, where $\Im(z)$ and $\Re(z)$ denote the imaginary and real parts of the complex number $z$ respectively. 
If we define the phase $\theta_j$ such that $w_j = \beta e^{i \theta_j}$, we can further express these as 
\begin{align}
2i \Im(\gamma^{-(i-1)(j-1)} w_j) & = 2i \beta \sin \left(\theta_j - \frac{2 \pi}{\kappa} (i-1)(j-1)\right)  
\end{align} 
and 
\begin{align} 
2 \Re(\gamma^{-(i-1)(j-1)} w_j) & = 2 \beta \cos \left(\theta_j - \frac{2 \pi}{\kappa} (i-1)(j-1) \right) . 
\end{align} 
To simplify our notation, we will define 
$$\tilde{\theta}_j := \theta_j - \frac{2 \pi}{\kappa} (i-1)(j-1), $$ 
allowing us to write the summation in (\ref{eqn:cwsummation0}) as 
$$\sum_{j = 2}^{\frac{\kappa}{2}} \left( \gamma^{-(i-1)(j-1)} w_j + \gamma^{-(i-1)((\kappa - j + 2) -1)} w_{\kappa - j + 2} \right)$$ 
$$ = \sum_{j \text{ even} } 2i \beta \sin(\tilde{\theta}_j) + \sum_{j \text{ odd}} 2 \beta \cos(\tilde{\theta}_j). $$


Now, we can bound the coherence by $$\mu \le \max_{\{\theta_j\}} \max_i |c_{x^{i-1}}| \le \max_i \max_{\{\theta_j\}} |c_{x^{i-1}}|, $$ and from our above discussion this becomes 
\begin{align} 
\max_{\{\tilde{\theta}_j\} } \frac{1}{\kappa} \left| - \frac{1}{m} \pm \beta + \sum_{j \text{ even} } 2i \beta \sin(\tilde{\theta}_j) + \sum_{j \text{ odd}} 2 \beta \cos(\tilde{\theta}_j) \right| \label{eqn:cohboundrdiv2even} 
\end{align} 
if $\frac{\kappa}{2}$ is even, and 
\begin{align} 
\max_{\{\tilde{\theta}_j\} } \frac{1}{\kappa} \left| - \frac{1}{m} \pm i \beta + \sum_{j \text{ even} } 2i \beta \sin(\tilde{\theta}_j) + \sum_{j \text{ odd}} 2 \beta \cos(\tilde{\theta}_j) \right| \label{eqn:cohboundrdiv2odd} 
\end{align} 
if $\frac{\kappa}{2}$ is odd. 

If we set $$n_e := \#\{j \text{ even}~|~ 2 \le j \le \frac{\kappa}{2} \}$$ and $$n_o := \#\{j \text{ odd}~|~ 2 \le j \le \frac{\kappa}{2} \}, $$ then by speculation, (\ref{eqn:cohboundrdiv2even}) becomes bounded by
\begin{align} 
\frac{1}{\kappa} \left|\frac{1}{m} + \beta +  n_e 2i \beta + n_o 2 \beta \right| = \frac{1}{\kappa} \sqrt{ \left(\frac{1}{m} + \beta (1 +  2 n_o) \right)^2  + (2 n_e \beta)^2 } \label{eqn:cohboundrdiv2even_2}
\end{align} 
and (\ref{eqn:cohboundrdiv2odd}) becomes bounded by 
\begin{align} 
\frac{1}{\kappa} \left| \frac{1}{m} + i \beta +  n_e 2i \beta + n_o 2 \beta \right| = \frac{1}{\kappa} \sqrt{ \left(\frac{1}{m} + 2 n_o \beta \right)^2  + \beta^2 (1 + 2 n_e)^2 }\label{eqn:cohboundrdiv2odd_2}
\end{align}

Finally, we note that when $\frac{\kappa}{2}$ is even, then $n_e = \frac{\kappa}{4}$ (half the numbers between 1 and $\frac{\kappa}{2}$, inclusive, are even), and hence $n_o = \left(\frac{r}{\kappa} - 1\right) - n_e = \frac{\kappa}{4} - 1$. 
Thus, (\ref{eqn:cohboundrdiv2even_2}) becomes 
\begin{align} 
\frac{1}{\kappa} \sqrt{ \left(\frac{1}{m} + \beta \left(1 +  2 \left(\frac{\kappa}{4} - 1 \right) \right) \right)^2  + \left(2 \cdot \frac{\kappa}{4} \cdot \beta\right)^2 } \\ 
= \frac{1}{\kappa} \sqrt{\left(\frac{1}{m} + \left(\frac{\kappa}{2} - 1 \right) \beta \right)^2 + \left(\frac{\kappa}{2}\right)^2 \beta^2}. 
\end{align} 
 
We get the same bound when $\frac{\kappa}{2}$ is odd. Indeed, in this case $n_e = \frac{1}{2} \left( \frac{\kappa}{2} - 1 \right) = \frac{\kappa}{4} - \frac{1}{2}$ (now half the numbers between $1$ and $\frac{\kappa}{2} - 1$, inclusive, are even), and $n_o = \left(\frac{\kappa}{2} - 1\right) - n_e = \frac{\kappa}{4} - \frac{1}{2}$. Then (\ref{eqn:cohboundrdiv2odd_2}) also becomes 
\begin{align} 
\frac{1}{\kappa} \sqrt{ \left(\frac{1}{m} + 2 \left(\frac{\kappa}{4} - \frac{1}{2}\right) \beta \right)^2  + \beta^2 \left(1 + 2 \left( \frac{\kappa}{4} - \frac{1}{2} \right) \right)^2 } \\ 
= \frac{1}{\kappa} \sqrt{\left(\frac{1}{m} + \left(\frac{\kappa}{2} - 1 \right) \beta \right)^2 + \left(\frac{\kappa}{2}\right)^2 \beta^2}. 
\end{align} 
 
This concludes the proof. 
\end{proof}

\noindent 
\textit{Remark:} If were to mimic the proof of Theorem \ref{thm:coherenceupperboundmodd} in the case when $m$ is even (so $-1 \in A$ and the $c_x^{d}$ are real), then we would arrive at the same bound as in Theorem \ref{thm:coherenceupperbound2}. Indeed, in this case from Lemma \ref{lem:cstarwstar} we have 
\begin{align} 
\left(\gamma^{-(i-1)(j-1)}w_j\right)^* = \gamma^{-(i-1)((\kappa-j+2)-1)} w_{\kappa-j+2}, \label{eqn:mevenwcondition} 
\end{align} 
for $j = 2, ..., \kappa$, and hence if $\kappa$ is odd we have 
\begin{align} 
c_{x^{i-1}} & = \frac{1}{\kappa} \sum_{j = 1}^{\kappa} \gamma^{-(i-1)(j-1)} w_j \\ 
& = \frac{1}{\kappa} \left[ w_1 + \sum_{j = 2}^{\frac{\kappa+1}{2}} \left( \gamma^{-(i-1)(j-1)} w_j + \gamma^{-(i-1)((\kappa-j+2)-1)} w_{\kappa - j + 2} \right) \right] \\ 
& = \frac{1}{\kappa} \left[ -\frac{1}{m} + \sum_{j = 2}^{\frac{\kappa+1}{2}} 2 \beta \cos(\tilde{\theta}_j) \right]. 
\end{align} 

If $\kappa$ is even, we note that our above condition (\ref{eqn:mevenwcondition}) implies that $\left(\gamma^{-(i-1)\frac{\kappa}{2} }w_{\frac{\kappa}{2} + 1} \right)^* = \gamma^{-(i-1)\frac{\kappa}{2}} w_{\frac{\kappa}{2} + 1}$, so we must have that $\gamma^{-(i-1)\frac{\kappa}{2}} w_{\frac{\kappa}{2} + 1}$ is real, and hence equal to $\pm \beta$. 
Thus we get 
\begin{align} 
c_{x^{i-1}} & = \frac{1}{\kappa} \left[ w_1 + \gamma^{-(i-1)\frac{\kappa}{2}} w_{\frac{\kappa}{2} + 1} + \sum_{j = 2}^{\frac{\kappa}{2}} \left( \gamma^{-(i-1)(j-1)} w_j + \gamma^{-(i-1)((\kappa-j+2)-1)} w_{\kappa - j + 2} \right) \right] \\
& = \frac{1}{\kappa} \left[ - \frac{1}{m} \pm \beta + \sum_{j = 2}^{\frac{\kappa}{2}} 2 \beta \cos( \tilde{\theta}_j) \right].  
\end{align} 

In either case, maximizing over $\{\theta_j\}$ gives us an upper bound of 
\begin{align} 
\mu \le \frac{1}{\kappa} \left( \frac{1}{m} + (\kappa-1)\beta \right), 
\end{align} 
which matches with our bound from Theorem \ref{thm:coherenceupperbound2}.

\end{appendices}


\end{document}